\documentclass[11pt]{article}
\usepackage[latin9]{inputenc}
\usepackage[letterpaper]{geometry}
\usepackage{xcolor}
\usepackage{float}
\usepackage{bm}
\usepackage{amsmath}
\usepackage{amsthm}
\usepackage{amssymb}
\usepackage{graphicx}
\usepackage{subfigure} 
\usepackage{esint}
\usepackage[authoryear,round]{natbib}
\usepackage{soul}
\makeatletter

\floatstyle{ruled}
\newfloat{algorithm}{tbp}{loa}
\providecommand{\algorithmname}{Algorithm}
\floatname{algorithm}{\protect\algorithmname}

\theoremstyle{plain}
\newtheorem{lem}{\protect\lemmaname}
\theoremstyle{plain}
\newtheorem{thm}{\protect\theoremname}
\theoremstyle{remark}

\theoremstyle{definition}
\newtheorem{example}{\protect\examplename}
\theoremstyle{plain}
\newtheorem{cor}{\protect\corollaryname}

\@ifundefined{date}{}{\date{}}

\usepackage{caption}
\usepackage{array}
\usepackage{longtable}
\usepackage{afterpage}
\usepackage{pdflscape}
\usepackage{arydshln}
\usepackage{multirow}
\usepackage{url}
\usepackage{amsthm}
\usepackage{bbm}
\usepackage{epstopdf}
\usepackage{setspace}
\usepackage{titlesec}
\usepackage[acronym,nonumberlist]{glossaries}
\usepackage{caption}
\usepackage{xcolor}
\usepackage[noend]{algpseudocode}
\usepackage{algorithm}
\usepackage{algorithmicx}
\usepackage{mdframed}

\usepackage{soul}

\makeatother
\providecommand{\corollaryname}{Corollary}
\providecommand{\examplename}{Example}
\providecommand{\lemmaname}{Lemma}
\providecommand{\remarkname}{Remark}
\providecommand{\theoremname}{Theorem}

\begin{document}
\begin{titlepage} 
\title{Targeting Kollo Skewness \\with Random Orthogonal Matrix Simulation}
\author{Carol Alexander\thanks{Corresponding author. University of Sussex Business School and Peking University PHBS Business School.   Email: c.alexander@sussex.ac.uk}
$\,$ Xiaochun Meng\thanks{University of Sussex Business School. Email:   xiaochun.meng@sussex.ac.uk.} \, Wei Wei\thanks{University of Sussex Business School. Email: wei.wei@sussex.ac.uk.} }

\maketitle
 
\begin{abstract}
\noindent High-dimensional multivariate systems often lack closed-form solutions and are therefore resolved using simulation. Random Orthogonal Matrix (ROM) simulation is a state-of-the-art method that has gained popularity in this context because certain simulation errors are completely removed \citep{ledermann2011random}. Specifically, every sample generated with ROM simulation \textit{exactly} matches a target mean and covariance matrix. A simple extension can also target a scalar measure of skewness or kurtosis. However, targeting non-scalar measures of higher moments is much more complex. {\color{blue} 
The problem of targetting exact Kollo skewness has already been considered, but the algorithm  proceeds via time-consuming trial-and-error which can be very slow. Moreover, the algorithm often fails completely. Furthermore,} it produces simulations with very long periods of inactivity which are inappropriate for most real-world applications. This paper provides an in-depth theoretical analysis of a much quicker ROM simulation extension which always succeeds to target Kollo skewness exactly. {\color{blue} It also derives new results on Kollo skewness in concatenated samples and applies them} to produce realistic simulations for many applications, especially those in finance. Our first contribution is to establish necessary and sufficient conditions for Kollo skewness targeting to be possible. Then we introduce two novel methods, one for speeding up the algorithm and another for improving the statistical properties of the simulated data. We  illustrate several new theoretical results with some extensive numerical analysis and we apply the algorithm using some real data drawn from two different multivariate systems of financial returns.

\bigskip{}
\bigskip{}

 \vspace{0in}
 \bigskip{}
\end{abstract}
\setcounter{page}{0}{} \thispagestyle{empty} \end{titlepage} \pagebreak\onehalfspacing
\pagenumbering{arabic} 

\section{Introduction}
Modelling multivariate systems is important for many applications in operational research but also in economics, engineering, environmental science, finance, medicine and other disciplines that require quantitative analysis. In the vast majority of cases, the multivariate distributions under scrutiny have no analytic or closed form, so the problem can only be resolved using some type of numerical technique. Of these, Monte Carlo (MC) simulation is most commonly used because it is fast, the basic method is quite straightforward and the random samples generated converge to the distribution in probability \citep{thomopoulos2012}. However, the sample moments are not equal to their population values and even when the experimenter performs a very large number of simulations the differences can be considerable. As a result, sampling error and uncertainty in MC methods have been widely studied (see, for example, \citealt{oberkampf2002} and \citealt{Ferraz2012}). For this reason, much research focusses on developing variance-reduction techniques aimed at improving the efficiency of simulations -- see, for example, \cite{Swain1988}, \cite{Hesterberg1996}, \cite{Szantai2000}, \cite{Saliby2009} and many others. These have applications to any problem which requires the use of multivariate simulations for its resolution. For instance, see \citet{Badano2013} and \citet{Quintana2015} for multivariate simulations used in engineering applications, and \citet{Avramidis2002}, \citet{Capriotti2008}, and \citet{Ma2016} for problems in finance.

Addressing the simulation error problem from a different perspective entirely, \citet{ledermann2011random} introduces a novel and flexible method for generating random samples from a correlated, multivariate system called Random Orthogonal Matrix (ROM) simulation. It is based on the premise of \textit{exactly} matching a target mean, covariance matrix and certain higher moments with \textit{every} simulation. Several developments have extended the original paper, for instance,
\citet{ledermann2012further} investigate the effect of different random rotation matrices on ROM simulation sample characteristics. And \citet{hurlimann2013generalized} presents a very significant theoretical development of ROM simulation by extending the basic $L$ matrices (which are fundamental to ROM simulation, as we shall see below) to a much broader class of generalised Helmert-Ledermann (GHL) matrices, thereby increasing flexibility and scope of ROM simulation.

ROM simulation has already been applied in a variety of real-world applications. For instance, \citet{geyer2014no} implements ROM simulation in an arbitrage-free setting, for option pricing and hedging. \citet{hurlimann2014market} integrates ROM simulation into two semi-parametric methods that take into account skewness and kurtosis, in order to improve the computation of quantiles for linear combinations of the system variables, and \citet{alexander2012} investigate the use of target skewness and kurtosis via ROM simulation for stress testing. Further theoretical developments focus on the higher-moment characteristics of ROM simulation: \citet{hurlimann2015mardia} introduces a new, recursive method for targeting the skewness and kurtosis metrics of \citet{mardia1970measures} using ROM simulation with GHL matrices, and \citet{hanke2017} improve on the skewness metric that ROM simulation can target, using the vector-valued moment of \citet{kollo2008multivariate} in place of the Mardia skewness.

This paper extends the literature on higher-moment characteristics of ROM simulation by providing a new method for targeting Kollo skewness. The \cite{hanke2017} algorithm requires selecting a large number of so-called `arbitrary' real values to determine a complex non-linear system of equations. However, the authors overlooked that these equations are often not solvable in the real numbers, an issue that becomes more and more prevalent as the dimension of the system increases.

Their algorithm entails a trial-and-error approach but we found that it never works in many cases, and when it does the  simulations have long periods of zero or minimal activity interspersed with periods of high activity. This could be a useful feature for applications to climate modelling or medical research -- for instance, such patterns are commonly observed in multivariate seismic wave dynamics (\citealt{tk1991}, \citealt{adelfio2012} and others) and cardiology  (see \citealt{Mcsharryetal2003},  \citealt{Lyonetal2018} and others). However, such extreme patterns have little relevance to problems in economics and finance.	

We  develop a novel ROM algorithm for targeting Kollo skewness, hereafter referred to as KROM, which utilizes a set of new, necessary and sufficient conditions for the existence of real solutions to the exact Kollo skewness equations. A numerical study demonstrates a very large reduction in computation time, and much lower failure rates {\color{blue} of finding arbitrary value compared with the trial-and-error approach in the existing algorithm.} Being based on bootstrapping real data or parametric distributions KROM simulation offers the flexibility to generate simulations with statistical characteristics suitable to many applications of interest.  It is ideally suited to modelling financial systems, which can be very high dimensional and have long been known to exhibit the type of conditional heteroscedasticity which arises naturally using the KROM algorithm.\footnote{See \citealt{engle1982arch}, \citealt{bollerslev1994arch} and a highly-prolific theoretical and applied literature since then.} An empirical study demonstrates its application to hourly cryptocurrency returns and daily returns on the sector indices of the S\&P 500 index.

Another new theoretical result is used in KROM simulation to provide the option of sample concatenation. It is already known that concatenating samples generated by standard ROM simulation preserves the target mean and covariance matrix \citep{ledermann2011random}, but here we prove that sample concatenation can also leave Kollo skewness invariant, under certain conditions. Thus, concatenation may be applied to increase the scope of KROM simulation to reflect characteristics of real-world applications, especially in large systems.  By contrast, the {\color{blue} existing} algorithm, when it works,  produces simulations having strange dynamic properties and exceptionally high marginal kurtosis which limits the scope of its potential applications.

 In the following: Section \ref{sec:basics} fixes notation and concepts; Section \ref{sec:ext} presents the new algorithm by deriving the necessary and sufficient conditions for real solutions of the Kollo skewness equations,  proposing both non-parametric and parametric methods for generating simulations, presenting comprehensive numerical  results which aim to guide the researcher to operate the algorithm by choosing appropriate parameters, and proving  the invariance of Kollo skewness under sample concatenation;  Section \ref{sec: real data} reports our practical applications using real financial systems; and Section \ref{sec:conclusion} summarises and concludes. All theoretical proofs and some further numerical results are given in the appendices and all Python and Matlab code is available from the authors on request.


\section{Background}\label{sec:basics}
In this section we fix notation by providing a brief overview of the basic concepts for both ROM simulation and multivariate skewness. 
\subsection{ROM Simulation}
Suppose we want to simulate from an $n$-dimensional random variable with $n \times 1$ mean vector $\bm{\mu}_{n}$  and $n \times n$ covariance matrix $\mathbf{V}_{n}$. Let $\mathbf{X}_{mn}$ denote the $m\times n$ matrix of simulations, where $m$ denotes the number of observations. In order to match the mean and covariance of the samples $\mathbf{X}_{mn}$ with $\bm{\mu}_{n}$ and $\mathbf{V}_{n}$, the following equations must be satisfied:
\begin{align}
m^{-1}\bm{1}_{m}' & \mathbf{X}_{mn}=\bm{\mu}_{n}',\label{eq:target_mean}\\
m^{-1}\left(\mathbf{X}_{mn}^{'}-\bm{\mu}_{n}\bm{1}_{m}^{'}\right)&\left(\mathbf{X}_{mn}-\bm{1}_{m}\bm{\mu}_{n}'\right)=\mathbf{V}_{n},\label{eq:target_cov}
\end{align}
where $\bm{1}_{m}$ is a $m \times 1$ column vector with all ones and $^{'}$ denotes the matrix transpose. \citet{ledermann2011random} show that \eqref{eq:target_mean} and \eqref{eq:target_cov} hold if and only if $\mathbf{X}_{mn}$ takes the following form: 
\begin{align}
\mathbf{X}_{mn}=\bm{1}_{m}\bm{\mu}_{n}'+ & m^{1/2}\mathbf{Q}_{m}\mathbf{L}_{mn}\mathbf{R}_{n}\mathbf{A}_{n},\label{eq:ROM}\\
\mathbf{L}_{mn}'\mathbf{L}_{mn} & =\,\mathbf{I}_{n},\label{eq:L matrix1}\\
\bm{1}_{m}'\mathbf{L}_{mn} & =\,\bm{0}_{n}',\label{eq:L matrix2}
\end{align}
where the $n \times n$ matrix $\mathbf{A}_n$ is derived using $\mathbf{V}_{n}=\mathbf{A}_{n}^{'}\mathbf{A}_{n}$, 
$\bm{0}_{n}$ is the $n \times 1$ vector with all zeros, $\mathbf{I}_{n}$ is the $n \times n$ identity matrix, and both $\mathbf{Q}_{m}$ and $\mathbf{R}_{n}$ are orthogonal matrices. In ROM simulation, $\mathbf{Q}_{m}$ is a $m \times m$ random permutation matrix which functions to re-order the observations; $\mathbf{R}_{n}$ is a $n \times n$ random rotation matrix which can re-generate $\mathbf{X}_{mn}$ by simply changing $\mathbf{R}_n$, and the matrix $\mathbf{L}_{mn}$ is taken from a class of rectangular $m \times n$ orthogonal matrices known as $L$ matrices. In equation \eqref{eq:ROM}, the sample mean and covariance of $\mathbf{X}_{mn}$ are determined solely by $\bm{\mu}_{n}$ and $\mathbf{A}_{n}$, given that $\mathbf{L}_{mn}$ satisfies equations \eqref{eq:L matrix1} and \eqref{eq:L matrix2}. Hence, $L$ matrices are a cornerstone of ROM simulation and to generate $\mathbf{X}_{mn}$ we must first obtain an $L$ matrix.

In addition to the ability to match the exact mean and covariance matrix, ROM simulation is generally much faster than MC simulation, especially in high-dimensional systems, because samples are generated by simply changing $\mathbf{R}_{n}$ to another random rotation matrix and performing the required matrix multiplications. See \citet{ledermann2012further} for the classification of different types of rotation matrix and their effects on ROM simulation sample characteristics. 

\subsection{Kollo Skewness}
	Consider a correlated system of $n$ random variables $\mathbf{Y}_{n}=\left[Y_{1},Y_{2},\cdots,Y_{n}\right]$
having mean $\bm{\mu}_{n}$ and covariance matrix $\mathbf{V}_{n}$ and use the notation $\mathbf{Y}_{n}^{*}$ for the standardised versions of $\mathbf{Y}_{n}$, i.e. 
\[
\mathbf{Y}_{n}^{*}=(\mathbf{Y}_{n}-\bm{1}_{m}\bm{\mu}_{n}^{'})\mathbf{V}_{n}^{-1/2}.
\]
A typical way to define the skewness of their multivariate distribution is to set $c_{ijk}=\mathbb{E}[Y_{i}^{*}Y_{j}^{*}Y_{k}^{*}]$ for each tuple $(Y_{i}^{*},Y_{j}^{*},Y_{k}^{*})$, with $i,j,k=1,\ldots,n$. Then, for each $i=1,\ldots,n$, form the matrix 
\[
\mathbf{C}_{i}=\begin{bmatrix}c_{i11} & c_{i12} & \cdots & c_{i1n}\\
\vdots & \vdots & \ddots & \vdots\\
c_{in1} & c_{in2} & \cdots & c_{inn}
\end{bmatrix},
\]
and finally define the {co-skewness tensor} as: 
\begin{equation}
\mathbf{m}_{3}=\left[\mathbf{C}_{1},\mathbf{C}_{2},...,\mathbf{C}_{n}\right].\label{eq:co-skewness}
\end{equation}
This tensor has good properties such as invariance under affine, non-singular transformations \citep{mardia1970measures}. However, it has $n^{3}$ elements which makes it impractical for use in large systems, hence the need to construct summary skewness metrics.

\citet{mardia1970measures} proposes multivariate skewness as follows:
\begin{equation}
\beta_{1}\left(\mathbf{Y}_{n}\right)=\sum_{i,j,k=1}^{n}c_{ijk}^{2}.
\label{app:eq:mardia-skeness}
\end{equation}
It can be shown that Mardia skewness is essentially the sum of all the squared elements in \eqref{eq:co-skewness}. However, \citet{gutjahr1999} show that, being a simple scalar, the Mardia skewness loses much of the information contained in the co-skewness tensor \eqref{eq:co-skewness}. To retain more information yet still employ a skewness measure that is more practical than the full co-skewness tensor, we shall focus on the $n$-dimensional skewness measure proposed by \citet{kollo2008multivariate} which is defined as: 
\begin{equation}
	\bm{\tau}\left(\mathbf{Y}_{n}\right)=
	\begin{bmatrix}\sum_{i,k}^{n}c_{i1k},
		\sum_{i,k}^{n}c_{i2k},
		\cdots,
		\sum_{i,k}^{n}c_{ink}
	\end{bmatrix}^{'}.\label{eq:kollo-skewness}
\end{equation}
That is,  the $i^{th}$ element of the Kollo skewness vector is equal to the sum of the elements of
the $i$-th row of $\mathbf{m}_3$. It takes information from all elements of the co-skewness tensor, whereas the other multivariate skewness measures proposed by \cite{mori1994} and \cite{balakrishnan2007} omit  third-order mixed moments. Because of this, Kollo skewness  performs better for asymmetric multivariate distributions, as demonstrated by  \cite{jammalamadaka2020}.

\subsection{The Algorithm of \cite{hanke2017}}
Let $\bm{\mu}_{n}$, $\mathbf{V}_{n}$ and $\bm{\tau}_{n}$ be the target mean vector, covariance matrix, and Kollo skewness vector that we want to match. Equations \eqref{eq:ROM} -- \eqref{eq:L matrix2} only guarantee that $\mathbf{X}_{mn}$ matches the target mean and covariance, but it remains to match Kollo skewness which can be affected by $\mathbf{L}_{mn}$ and $\mathbf{R}_{n}$. We seek a suitable $m^{-1/2}\mathbf{L}_{mn}\mathbf{R}_{n}$, a standardised version of $\mathbf{X}_{mn}$, which also satisfies equation \eqref{eq:kollo-skewness}. 
	
	For convenience, we now re-state the \citet{hanke2017} algorithm in a different form because this clarifies the solvability issues that we wish to focus on. To this end, define the scaled $L$ matrix  $m^{1/2}\mathbf{L}_{mn}\mathbf{R}_{mn}=\mathbf{M}_{mn}=\mathbf{S}_{mn}\mathbf{\Omega}_{n}^{'}$ where $\mathbf{S}_{mn}$ is also a scaled $L$ matrix and $\mathbf{\Omega}_{n}$ is any orthogonal rotation matrix whose first column is $n^{-1/2}\mathbf{1}_{n}$ having the property that:
	\begin{equation}\label{omega}
	\mathbf{1}_{n}^{'}\mathbf{\Omega}_{n}=n^{1/2}\left[1,0,\dots,0\right].
	\end{equation}	
As a consequence, the algorithm aims to find $\mathbf{S}_{mn}$ and then obtain $\mathbf{X}_{mn}$ using:
	\begin{equation}
	\mathbf{X}_{mn}=\bm{1}_{m}\bm{\mu}_{n}'+\mathbf{Q}_{m}\mathbf{S}_{mn}\bm{\Omega}_{n}^{'}\mathbf{A}_{n}.\label{eq: get X from S}
	\end{equation}	
According to the definition, to solve the equations $\bm{\tau}\left(\mathbf{X}_{mn}\right)=\bm{\tau}_{n}$ for some target Kollo skewness vector $\bm{\tau}_{n}$ the scaled $L$ matrix should satisfy:
	\begin{equation}
	\bm{\tau}_{n}=m^{-1}\sum_{i=1}^{m}\left(\mathbf{r}_{i}\mathbf{1}_{n}\right)^{2}\mathbf{r}_{i}^{'}=m^{-1}\sum_{i=1}^{m}\left(\mathbf{s}_{i}\mathbf{\Omega}_{n}^{'}\mathbf{1}_{n}\right)^{2}\mathbf{\Omega}_{n}\mathbf{s}_{i}^{'}=m^{-1}n\,\mathbf{\Omega}_{n}\sum_{i=1}^{m}s_{i1}^{2}\mathbf{s}_{i}^{'}.\label{eq:simplified M_Kool}
	\end{equation}
	where $\mathbf{r}_{i}$ is the $i$-{th} row of $\mathbf{M}_{mn}$ and $\mathbf{s}_{i}$ is the $i$-{th} row of $\mathbf{S}_{mn}$,
	$i=1,\ldots,m$. Pre-multiplying the matrix $n^{-1}\mathbf{\Omega}_{n}^{'}$ on
	both sides of equation \eqref{eq:simplified M_Kool} the system may now be written: 
	\begin{align}
	\qquad\qquad\qquad\mathbf{p}_{n} & =m^{-1}\sum_{i=1}^{m}s_{i1}^{2}\mathbf{s}_{i}^{'}, & \text{(Exact Kollo skewness)}\label{eq:p}\\
	\intertext{where the rotated Kollo skewness vector  $\mathbf{p}_{n}=n^{-1}\mathbf{\Omega}_{n}^{'}\bm{\tau}_{n}$. In addition, because \ensuremath{\mathbf{S}_{mn}} is a scaled \ensuremath{L} matrix, it must satisfy: }
	\mathbf{S}_{mn}^{'}\mathbf{S}_{mn} & =m\mathbf{I}_{n}, & \text{(Exact covariance matrix)}\label{eq:S1}\\
	\mathbf{1}_{m}^{'}\mathbf{S}_{mn} & =\mathbf{0}_{n}^{'}. & \text{(Column sum constraint)}\label{eq:S2}
	\end{align}

The above formulation of the exact Kollo skewness equations highlights three necessary and sufficient conditions  \eqref{eq:p} -- \eqref{eq:S2} for $\mathbf{S}_{mn}$ to realise a ROM simulation with target moments $\bm{\mu}_{n}$, $\mathbf{V}_{n}$ and $\bm{\tau}_{n}$ via equation \eqref{eq: get X from S}. However, this formulation does not identify how these equations can be solved. To this end, let $s_{1k},\ldots,s_{mk}$ denote the elements in the $k$-{th} column
of $\mathbf{S}_{mn}$ and $p_{k}$ denote the $k$-{th} element of $\mathbf{p}_{n}$. Note that our theoretical results in the next section will highlight the importance of the first element of $\mathbf{p}_{n}$ given by  $p_1=n^{-3/2}\mathbf{1}_{n}^{'}\bm{\tau}_{n}$. With this notation, for each column $k$, we have an alternative expression for equations \eqref{eq:p} -- \eqref{eq:S2}
 as follows:
\begin{subequations}
\renewcommand{\theequation}{\theparentequation-\arabic{equation}}
\begin{align}
m^{-1}\sum_{i=1}^{m}\left(s_{i1}\right)^{2}s_{ik} & =p_{k}, & \text{(Exact Kollo skewness)}\label{eq:kollo:p}\\
m^{-1}\sum_{i=1}^{m}s_{ik}^{2} & =1, & \text{(Exact covariance matrix 1)}\label{eq:kollo:c1}\\
\sum_{i=1}^{m}s_{ik}s_{ij} & =0,\quad j<k, & \text{(Exact covariance matrix 2)}\label{eq:kollo:c2}\\
\sum_{i=1}^{m}s_{ik} & =0. & \text{(Column sum constraint)}\label{eq:kollo:sum}
\end{align}
\end{subequations}

 We need to solve equations \eqref{eq:kollo:p} to \eqref{eq:kollo:sum} sequentially, from $k=1$ to $k=n$. However, there are $m$ unknown variables in each column but  there are only $k+2$ constraints in the $k^{th}$ column. For example, for $k=1$ the $m$ unknowns are $s_{11},s_{21},\dots,s_{m1}$, but there are only 3 constraints. {\color{blue} Prescribing arbitrary values $w_{ik}$ to some of the unknowns can address this under-determinacy problem.} Specifically, 
the last $m-(k+2)$ elements in each column $k$ are set to arbitrary values, i.e.   $s_{ik}= w_{ik}$ for $i=k+3,...,m$ and for $k = 1, ...,n$. In the code supplied they set $w_{i1} =0 $ for $i=4,...,m$, and for $k = 2, ..., n$ they draw $w_{ik}$ randomly from a uniform distribution on $[0,1]$. This construction forces the exact  Kollo skewness structure to be derived entirely from the first few elements in each column of $\mathbf{S}_{mn}$. That is, the vast majority of elements in each simulation  have no effect on the value of Kollo skewness. 

Another problem arises from the following statement in \citet{hanke2017}: ``In this procedure, arbitrary values can be prescribed on several occasions. By varying these values, e.g. by setting them as random values, different $L$ matrices can be constructed, all of which exactly match the pre-specified Kollo skewness vector." However, equations \eqref{eq:kollo:p} to \eqref{eq:kollo:sum} cannot always be solved. That is, their so-called `arbitrary' values are not always admissible for this system of non-linear equations to have solutions in the real numbers. This may be seen immediately, for example, by rewriting  equation \eqref{eq:kollo:c1} as 
	\[
	s_{1k}^{2}+\dots+s_{k+2,k}^{2}=m-\sum_{i=k+3}^{m}w_{ik}^{2},	
	\]
	which must clearly have complex roots when the arbitrary values $w_{ij}$ are such that $\sum_{i=j+3}^{m}w_{ij}^{2}>m$. 
	
Yet another complication arises because the additional 
	restriction that each simulation has exact  Kollo skewness dramatically increases the complexity of the computations, so it is very much slower than the original ROM simulation algorithm of \cite{ledermann2011random}. 
	For clarity, we summarise the main weaknesses of the algorithm as follows: 
	\begin{enumerate}
			\item The statistical characteristics of the samples generated by the  algorithm are unsuited to many applications because the vast majority of elements in the simulation are arbitrary values that are either zero or have minimal variation;
		\item There is no guarantee that the algorithm  works, and when it does work  it can be very slow, which further reduces its appeal for practical implementation.
	
	\end{enumerate}

\section{A New ROM Algorithm to Target Kollo Skewness}\label{sec:ext} 
On closer inspection an algorithm that targets Kollo skewness could be made very much faster if we could find a way to minimize the number of trials taken to produce a set of  arbitrary values for which \eqref{eq:kollo:p} to \eqref{eq:kollo:sum} actually has real solutions.  This motivates our first theoretical result, Theorem 1 in Section \ref{sec:nece and suff conditions}, which proves the necessary and sufficient conditions for the arbitrary values to be admissible, thereby circumventing weakness (2) above. Then, to overcome weakness (1) as well, Section  \ref{sec:aav} proposes a new ROM simulation algorithm with target Kollo skewness, or KROM simulation for short, which utilizes Theorem 1 in the context of a statistical bootstrap or a parametric distribution.  We also establish a practically useful theoretical result that allows the option to use sample concatenation in KROM simulation, i.e. to join several small sub-samples into one simulation. Finally, Section \ref{numerical results}  presents some numerical results { on the  computational speed of our algorithm}, demonstrates how failure rates vary with both system and sample sizes, and with the variance of the arbitrary values used to solve the Kollo skewness equations, and finally exhibits the effect of using the concatenation option. 

\subsection{Solving the Kollo Skewness Equations}\label{sec:nece and suff conditions}
Suppose that the last $m-(k+2)$ elements in each column $k$ are set to arbitrary values and consider the necessary and sufficient conditions for these arbitrary values to be admissible. Equations \eqref{eq:kollo:p} to \eqref{eq:kollo:sum} must be solved iteratively, for $j=1,2,\dots,n$. First consider the case $j=1$. Setting $s_{ij}=w_{ij}$ for $i=4,...,m$ we can simplify and rewrite the system \eqref{eq:kollo:p} to \eqref{eq:kollo:sum} as the following three equations: 
\begin{subequations}
	\renewcommand{\theequation}{\theparentequation-\arabic{equation}}
	\begin{align}
	s_{11}+s_{21}+s_{31} & =-\sum_{i=4}^{m}w_{i1}=:a\label{eq:system1}\\
	s_{11}^{2}+s_{21}^{2}+s_{31}^{2} & =m-\sum_{i=4}^{m}w_{i1}^{2}=:b\label{eq:system2}\\
	s_{11}^{3}+s_{21}^{3}+s_{31}^{3} & =mp_{1}-\sum_{i=4}^{m}w_{i1}^{3}=:c.\label{eq:system3}
	\end{align}
\end{subequations}
	
\begin{lem}
	\label{sec:iff}  \label{lemma 1}
A necessary and sufficient condition for the system \eqref{eq:system1}
	to \eqref{eq:system3} to have real solutions for $s_{11},s_{21}$
	and $s_{31}$ is:
	\begin{equation}
	\left(b-\frac{a^{2}}{3}\right)^{3}\ge6\left(c-ab+\frac{9}{2}a^{3}\right)^{2}.\label{eq:first_realroots}
	\end{equation}
\end{lem}
\vspace{12pt}
	
\noindent The proof is in  Appendix A. Now consider the case when $j=2,\dots,n$. The system \eqref{eq:kollo:p} to \eqref{eq:kollo:sum} is solved iteratively, so we assume it is solvable in the real numbers for $j=1,..,k-1$, and then seek necessary and sufficient conditions for the system to be solvable when $j=k$. In this case, the system \eqref{eq:kollo:p} to \eqref{eq:kollo:sum} is most succinctly expressed in matrix form, as: 
\begin{align}
\mathbf{U}\mathbf{y} & =\mathbf{v},\label{eq:linear_equations}\\
\mathbf{y}^{'}\mathbf{y} & =m-\sum_{i=k+3}^{m}w_{ik}^{2}.\label{eq:quadratic_equation}
\end{align}
where $\mathbf{y}$ is a column vector containing all the variables
we want to solve, i.e. $(s_{1,k},\dots,s_{k+2,k})^{'}$ and
 
\[
\mathbf{U}=\begin{bmatrix}s_{11}^{2} & s_{21}^{2} & \dots & s_{k+2,1}^{2}\\
s_{11} & s_{21} & \dots & s_{k+2,1}\\
s_{12} & s_{22} & \dots & s_{k+2,2}\\
\vdots & \vdots & \vdots & \vdots\\
s_{1,k-1} & s_{2,k-1} & \dots & s_{k+2,,k-1}\\
1 & 1 & 1 & 1
\end{bmatrix},\mathbf{v}=\begin{bmatrix}mp_{k}-\sum_{i=k+3}^{m}s_{i1}^{2}w_{ik}\\
-\sum_{i=k+3}^{m}s_{i1}w_{ik}\\
-\sum_{i=k+3}^{m}s_{i2}w_{ik}\\
\dots\\
-\sum_{i=k+3}^{m}s_{i,k-1}w_{ik}\\
-\sum_{i=k+3}^{m}w_{ik}
\end{bmatrix}.
\]
\\
Thus $\mathbf{U}$ is a $(k+1)\times(k+2)$ coefficient matrix and
$\mathbf{v}$ is a column vector. In practical applications we have
always found that the column span of $\mathbf{U}$ equals $k+1$,
i.e. the matrix is full row rank. However, in some applications it
could be that they are not. So, let $\mathbf{U}_{1}$ be a matrix
containing all the linearly independent columns of $\mathbf{U}$,
with the remaining columns in the matrix $\mathbf{U}_{2}$. Denote by $\mathbf{U}_{1}^{+}$ the Moore-Penrose inverse of $\mathbf{U}_{1}$
which is defined as $\mathbf{U}_{1}^{+}=\left(\mathbf{U}_{1}^{'}\mathbf{U}_{1}\right)^{-1}\mathbf{U}_{1}^{'}$. 
  
\begin{thm}
	\label{theorem HPSW} A real solution of \eqref{eq:kollo:p} -- \eqref{eq:kollo:sum} 
	exists iff the following three conditions hold: 
	
(i) The arbitrary values selected in the first column of $\mathbf{S}_{mn}$
		satisfy \eqref{eq:system1} -- \eqref{eq:system3} and \eqref{eq:first_realroots};

(ii) When admissible values exist for columns $1$
to $k-1$ then the arbitrary values in the $k$-{th}
column $(k=2,\dots,n)$ of $\mathbf{S}_{mn}$  have 	$
		\mbox{Rank}\left(\mathbf{U}\right)=\mbox{Rank}\left(\left[\mathbf{U},\mathbf{v}\right]\right)
	$
		where $\left[\mathbf{U},\mathbf{v}\right]$ is the augmented matrix;
		
		(iii) The arbitrary values selected in the $k$-{th}
		column $(k=2,\dots,n)$ of $\mathbf{S}_{mn}$ also have $$ 
		\mathbf{g}^{'}\mathbf{G}^{-1}\mathbf{g}-\mathbf{v}^{'}\left(\mathbf{U}_{1}\mathbf{U}_{1}^{'}\right)^{+}\mathbf{v}\,\ge\sum_{j=i+3}^{m}w_{ji}^{2}-m
$$
		where 
		\[
		\mathbf{G}=\mathbf{I}+\mathbf{U}_{2}^{'}\left(\mathbf{U}_{1}\mathbf{U}_{1}^{'}\right)^{+}\mathbf{U}_{2}\qquad\mbox{and}\qquad\mathbf{g}=\mathbf{U}_{2}^{'}\left(\mathbf{U}_{1}\mathbf{U}_{1}^{'}\right)^{+}\mathbf{v}.
		\]

\end{thm}

\noindent  Although complex theoretically, this result is easy to code. It improves on the computational efficiency substantially, because the three conditions are easy to check.  Without this result one must run a complex, large dimensional non-linear optimization algorithm numerous times, to search by trial-and-error for real roots of a very large number of  equations, which may not even exist. So, if the algorithm works, it is rather slow.

Recall that one weakness  of the { \cite{hanke2017}} algorithm is that zero arbitrary values produce simulations with very long periods with no activity, and this seriously limits their practical applications. Yet another problem with zero arbitrary values is that they are not always admissible, so the algorithm may not work. To see this, we derive a simple corollary to  Theorem 1 which is proved in Appendix A:
\begin{cor}\label{lem:zero arbitrary values}
	Zero arbitrary values $w_{ik}=0$ are admissible if and only if 
		\begin{equation*}
			p_{1}^{2}\le\frac{m}{6} \qquad \mbox{and} \qquad
		p_{k}^{2} \le \frac{t}{m}+\frac{m}{(k+1)(k+1)}  \qquad \forall  k>1,
		\end{equation*}
		where $t=s_{11}^{4}+s_{21}^{4}+s_{31}^{4}-m\left(\frac{m}{k+1}+p_{1}^{2}+\dots+p_{k-1}^{2}\right)>0$ 
\end{cor}

\noindent Hence, the algorithm is more likely to fail with zero arbitrary values when an element $p_k$ of the rotated Kollo skewness is large and the sample size $m$ is small.\label{Remark: zero arbitrary values}

\subsection{The KROM Algorithm}\label{sec:aav}
The conditions of Corollary \ref{lem:zero arbitrary values} are less strict than the conditions of Theorem \ref{theorem HPSW} with non-zero arbitrary values. However, we do not advise using zero values in practice, because then it is not only the Kollo skewness, but also the mean and covariance, that are derived entirely from the first $k+2$ elements in the $k$-th column, for $k=1, \ldots, n$. This can cause impractical distortions in the simulations, such as excessively high kurtosis and long periods of inactivity. Clearly, we need to be more careful about solving the indeterminate system of Kollo skews equations, so here we propose {two novel} approaches: using a statistical bootstrap, and using parametric distributions.

Suppose the target moments are derived from a sample $\mathbf{X}_{ln}^{*}$ where $l$ is the number of observations and $n$ is the number of variables -- we distinguish this  from  $\mathbf{X}_{mn}$ i.e. the samples generated using ROM simulation. We derive the corresponding scaled $L$ matrix using a rotation matrix $\mathbf{\Omega}_{n}$ as: 

\begin{equation} \label{eq:get S from X}
\mathbf{S}_{ln}^{*}=\left(\mathbf{X}_{ln}^{*}-\mathbf{1}_{l}\bm{\mu}_{n}^{'}\right)\mathbf{A}_{n}^{-1}\mathbf{\Omega}_{n}.
\end{equation}
For each column $k$ we bootstrap $m-(k+2)$  values from the corresponding column of $\mathbf{S}_{ln}^{*}$ and we denote these bootstrapped values $z_{ik}$, $i=k+3, \ldots, m$. Note that this does not impose any restriction on the dependency between columns, so we can simply do this bootstrap independently, column by column. This is the simplest approach but it seems desirable anyway because the columns of $\mathbf{S}_{mn}$ should have zero correlation. 
	
Based on these values, the next task is to solve equations \eqref{eq:kollo:p} to \eqref{eq:kollo:sum} to obtain columns with exactly zero mean, unit standard deviation and zero pair-wise correlation -- as well as exact Kollo skewness. For each bootstrapped sample we need to verify a solution to the simultaneous equations  \eqref{eq:kollo:p} to \eqref{eq:kollo:sum}. Now, 
since each column of $\mathbf{S}_{ln}^{*}$ has zero mean and unit standard deviation, it may well be that $\sum_{i=k+3}^{m}z_{ik}^2 \ge m$ and in this case  equation \eqref{eq:kollo:c1} has no real roots. Therefore, we set $$\bar{z}\ensuremath{=\frac{1}{m-(k+2)}\sum_{i=k+3}^{m}z_{ik}} \qquad \text{and} \qquad   s_{z}^{2}=\frac{1}{m-(k+2)}\sum_{i=k+3}^{m}(z_{ik}-\bar{z})^{2},$$ and 
transform the bootstrapped values $z_{ik}$ to $w_{ik}$ given by: 
\begin{equation}
w_{ik}=\sigma \left(  \frac{z_{ik}-\bar{z}}   {s_{z}}  \right),\label{eq:adjust}\qquad i=k+3,\dots,m.
\end{equation}
Now, $\sigma \in (0,1)$ is set by the user to ensure that $\sum_{i=k+3}^{m}z_{ik}^2 < m$. Lower values of $\sigma$ are more likely to result in real roots for \eqref{eq:kollo:c1} although no positive target value of $\sigma$ is guaranteed to work.  There is a trade off between setting a high $\sigma$ to maximize sampling variation in the bootstrap and a low $\sigma$ which reduces the failure rate. We still need to solve the whole system \eqref{eq:kollo:p} to \eqref{eq:kollo:sum}  and by virtue of simulation there will always be some values for which there is no solution. See Section \ref{sec: real data} for an analysis of various judicious values for $\sigma$  in an empirical application of failure rates.

As an alternative to  bootstrapping we may select the values from some representative parametric distributions, such as may be required in financial applications like risk management and portfolio optimization.  This way, we combine univariate MC simulation with solving equations  \eqref{eq:kollo:p} to \eqref{eq:kollo:sum}. As with the bootstrap, we may draw columns independently from any set of univariate marginal distributions which -- if we wish to exploit the maximum distributional flexibility of ROM simulation -- can be chosen to have very  different statistical properties. Nevertheless, when we do an empirical comparison of  failure rates in Section \ref{sec: real data} we shall only consider values that are drawn from the same  distribution family for each marginal.

Sample concatenation refers to the idea of ``stitching'' a finite number of samples to form a larger sample.  In mathematical notation, let $\mathbf{M}^{(k)}$, $k=1,2,...,N$ be $N$ sample matrices, each with dimension $m_{k}\times n$, then the sample concatenation of $\mathbf{M}^{(k)}$ leads to a larger sample $\begin{bmatrix} \mathbf{M}^{(1)'}, \cdots, \mathbf{M}^{(N)'} \end{bmatrix}^{'}$. \citet{ledermann2011random} first apply the sample concatenation method in ROM simulation. They show that the mean and covariance matrix of the concatenated samples are preserved if each constituent has identical mean and covariance matrix. Our next theoretical result shows how to ensure that Kollo skewness is also preserved under sample concatenation:

\begin{thm}
	\label{theorem concatenation 1} Let $\mathbf{M}^{(k)}$ be a scaled
	$L$ matrix of dimension $m_{k}\times n$ with $m_{k}\ge n+2$. Denote
	the $i^{th}$ row of $\mathbf{M}^{(k)}$ by $\mathbf{r}_{i}^{(k)}$,
	for $i=1,\dots,N$ where $N$ is the number of concatenations. Then
	the Kollo skewness of the concatenated sample is the weighted sum
	of their Kollo skewnesses, i.e. 
	\begin{equation}
	\begin{split}\bm{\tau}\left(\begin{bmatrix}\mathbf{M}^{(1)}\\
	\vdots\\
	\mathbf{M}^{(N)}
	\end{bmatrix}\right) & =\frac{1}{m_{1}+\dots+m_{N}}\left(\sum_{i=1}^{m_{1}}\left(\mathbf{r}_{i}^{(1)}\mathbf{1}_{n}\right)^{2}\mathbf{r}_{i}^{(1)'}+\dots+\sum_{i=1}^{m_{N}}\left(\mathbf{r}_{i}^{(N)}\mathbf{1}_{n}\right)^{2}\mathbf{r}_{i}^{(N)'}\right)\\
	& =\frac{1}{m_{1}+\ldots+m_{N}}\left(m_{1}\bm{\tau}\left(\mathbf{M}^{(1)}\right)+\dots+m_{N}\bm{\tau}\left(\mathbf{M}^{(N)}\right)\right),
	\end{split}
	\label{eq:concatenated1}
	\end{equation}
\end{thm}
\noindent  It follows from Theorem \ref{theorem concatenation 1} that Kollo skewness is invariant under sample concatenation, provided it is achieved via a set of (scaled) $L$ matrices having identical dimensions and the same Kollo skewness.

KROM simulation requires the rotation matrix  $\mathbf{\Omega}_{n}$ to satisfy equation \eqref{omega}. This is the one we use for the numerical results in Section \ref{numerical results}:
\begin{equation*}
\begin{pmatrix}\frac{1}{\sqrt{n}} & \frac{(-1)^{n-1}}{\sqrt{n/(n-1)}} & 0 & 0 & \cdots & 0 & 0 & 0\\
\frac{1}{\sqrt{n}} & \frac{(-1)^{n}}{\sqrt{n(n-1)}} & \frac{(-1)^{n-1}}{\sqrt{(n-1)/(n-2)}} & 0 & \cdots & 0 & 0 & 0\\
\frac{1}{\sqrt{n}} & \frac{(-1)^{n}}{\sqrt{n(n-1)}} & \frac{(-1)^{n}}{\sqrt{(n-1)(n-2)}} & \frac{(-1)^{n-1}}{\sqrt{(n-2)/(n-3)}} & \cdots & 0 & 0 & 0\\
\frac{1}{\sqrt{n}} & \frac{(-1)^{n}}{\sqrt{n(n-1)}} & \frac{(-1)^{n}}{\sqrt{(n-1)(n-2)}} & \frac{(-1)^{n}}{\sqrt{(n-2)(n-3)}} & \cdots & 0 & 0 & 0\\
\vdots & \vdots & \vdots & \vdots & \vdots & \vdots & \vdots & \vdots\\
\frac{1}{\sqrt{n}} & \frac{(-1)^{n}}{\sqrt{n(n-1)}} & \frac{(-1)^{n}}{\sqrt{(n-1)(n-2)}} & \frac{(-1)^{n}}{\sqrt{(n-2)(n-3)}} & \cdots & \dfrac{(-1)^{n}}{\sqrt{4\times3}} & \dfrac{(-1)^{n-1}}{\sqrt{3/2}} & 0\\
\frac{1}{\sqrt{n}} & \frac{(-1)^{n}}{\sqrt{n(n-1)}} & \frac{(-1)^{n}}{\sqrt{(n-1)(n-2)}} & \frac{(-1)^{n}}{\sqrt{(n-2)(n-3)}} & \cdots & \dfrac{(-1)^{n}}{\sqrt{4\times3}} & \dfrac{(-1)^{n}}{\sqrt{3\times2}} & \dfrac{(-1)^{n-1}}{\sqrt{2}}\\
\frac{1}{\sqrt{n}} & \frac{(-1)^{n}}{\sqrt{n(n-1)}} & \frac{(-1)^{n}}{\sqrt{(n-1)(n-2)}} & \frac{(-1)^{n}}{\sqrt{(n-2)(n-3)}} & \cdots & \dfrac{(-1)^{n}}{\sqrt{4\times3}} & \dfrac{(-1)^{n}}{\sqrt{3\times2}} & \dfrac{(-1)^{n}}{\sqrt{2}}
\end{pmatrix}.\label{Omega}
\end{equation*}
 
We  may now summarize the general KROM algorithm, with the option of sample concatenation,  in the following pseudo-code:

\begin{algorithm}[H]
	\footnotesize
	\caption{ KROM algorithm}
	\label{alg:modified algorithm-label}
	\textbf{The main code, with subsample concatenation option:}
	\begin{algorithmic}[1] 
		\Function{KROM}{$m,\bm{\mu},\mathbf{V},\bm{\tau}, N$}  \qquad	\Comment{$m$: sample size,\, $\bm{\mu},\mathbf{V},\bm{\tau}$: target moments, $N$: subsample size}
		
		\State $l=\max\left\lbrace n+2,\left\lfloor {\frac{m}{N}}\right\rfloor \right\rbrace $  \qquad \Comment{$\left\lfloor {x}\right\rfloor $ returns the greatest integer less than or equal to $x$} 
		
		\Comment{obtain the first $N-1$ subsamples with length $l$}
		\For{$k=1,\dots,N-1$} 
		\State $\mathbf{M}_{(kl-k+1):kl,n}$=\textsc{KROM\_GETM}$(l,\bm{\tau})$
		\EndFor 
		
		\Comment{obtain the last sub-samples with length $m-l(N-1)$} 
		\State $\mathbf{M}_{(Nl-l+1):m,n}$=\textsc{GETM}$(m-Nl-l,\bm{\tau})$
		
		\State obtain $\mathbf{A}_{n}$  \Comment{$\mathbf{A}^{'}\mathbf{A}=\mathbf{V}$}
		\State generate $\mathbf{Q}_m$     \Comment{ $\mathbf{Q}_m$ is a $m \times m$ random permutation matrix}
		\State $\mathbf{X}_{mn}=\mathbf{1}_{m}\bm{\mu}+ \mathbf{Q}_{m}\mathbf{M}\mathbf{A}_{n}$ \Comment{rotate back to get simulation $\mathbf{X}_{mn}$}
		\State \textbf{return} $\mathbf{X}_{mn}$
		\EndFunction
		\end{algorithmic}

		\textbf{\\A function to generate the matrix: $\mathbf{M}_{mn}$}
		\begin{algorithmic}[1] 
		\Function{GETM}{$m,\bm{\tau}$}   \Comment{$m$: the number of observations,\, $\bm{\tau}$: target Kollo skewness}
		\State $n=length(\bm{\tau})$
		\State generate $\mathbf{\Omega}_n$ \Comment{$\mathbf{\Omega}_n$  is a $n \times n$ rotation matrix satisfying equation~\eqref{omega}}
		\State $\mathbf{p}_n= \mathbf{\Omega}_{n}^{-1} n^{-1} \bm{\tau}$	
		    \Comment{obtain the rotated Kollo skewness}
		    
					\Comment{solve for the first column}
	
		\While{Theorem \ref{theorem HPSW} (i) does not hold}   

		\State generate random values $\left\lbrace z_{41},\dots,z_{m1}\right\rbrace$ by bootstrapping or from a distribution
		\State obtain $\left\lbrace w_{41},\dots,w_{m1}\right\rbrace$ using \eqref{eq:adjust} and check Theorem \ref{theorem HPSW} (i)
		\If{Theorem \ref{theorem HPSW} (i) holds}
		\State fill $\left\lbrace s_{41},\dots,s_{m1}\right\rbrace$ with $\left\lbrace w_{41},\dots,w_{m1}\right\rbrace$
		\EndIf
		\EndWhile
		\State solve \eqref{eq:kollo:p} to \eqref{eq:kollo:sum} for elements $\left\lbrace s_{11},\dots, s_{31} \right\rbrace $ 
			\Comment{solve for rotated Kollo skewness column by column}
		\For{$k=2, \dots, n$}       
		\While{Theorem \ref{theorem HPSW}(ii) and (iii) do not hold}  
		\State generate random values $\left\lbrace z_{k+3\,k},\dots,z_{mk}\right\rbrace$ by bootstrapping or from a distribution 
		\State obtain  $\left\lbrace w_{k+3\,k},\dots,w_{mk}\right\rbrace$ using \eqref{eq:adjust} and check Theorem~\ref{theorem HPSW} (ii) and  (iii)
		\If{Theorem~\ref{theorem HPSW} (ii) and (iii) hold}
		\State fill $\left\lbrace s_{k+3\,k},\dots,s_{mk}\right\rbrace$ with $\left\lbrace w_{k+3\,k},\dots,w_{mk}\right\rbrace$
		\EndIf
		\EndWhile
		\State solve \eqref{eq:kollo:p} to \eqref{eq:kollo:sum} to obtain elements $\left\lbrace s_{1k},\dots, s_{k+2\,k} \right\rbrace $
		\EndFor
		\State $\mathbf{M}_{mn}= \mathbf{S}_{mn}\mathbf{\Omega}_{n}^{'}$ \Comment{rotate back to get $\mathbf{M}_{mn}$}
		\State \textbf{return} $\mathbf{M}_{mn}$
		\EndFunction
		
	\end{algorithmic}
\end{algorithm}  

To generate the matrix $\mathbf{M}_{mn}$ one may select any parametric distribution, the choice is not constrained. However, when making this choice the researcher should bear two points in mind. First, the skewness that one can attain this way may be restricted, depending on the family selected. For instance, the skewness is zero when a Student's t distribution is used. In Appendix \ref{app: attainable p1} we derive the attainable skewness (under zero mean and unit variance) for some other well-known distributions: the skew-normal, the normal inverse Gaussian and the beta distribution. Another point to note is that the choice of distribution may affect the failure rates of the algorithm because the Kollo skewness equations need not be solvable. For instance, an obvious candidate for the parametric distribution is a normal distribution but this is highly likely to fail unless the target Kollo skewness is the  zero vector. We derive some additional theoretical results in Appendix B for this normal case, showing that the $\sigma$ resulting in admissible values increases monotonically with any pre-specified failure rate  $\alpha$. The proof, which is essentially based on the central limit theorem, is lengthy and tedious. To prevent the theoretical treatment becoming too overwhelming, we do not discuss this result further in the main text, but refer  readers to  Appendix B.



\subsection{Numerical Results}\label{numerical results}

Next we provide a simulation study which demonstrates the computational efficiency gained using our proposed KROM algorithm and demonstrates the effect of the concatenation option. First, Table \ref{tab:addlabel} compares the average computation time required to implement KROM simulations with those produced using the { \cite{hanke2017}} algorithm. For each value of $n$ and $m$, we randomly generate 100 target Kollo skewness vectors with each element chosen from the uniform distribution $\mathcal{U}(-1,1)$. For each of these 100 target Kollo skewness vectors, we note the  time required to implement the both algorithms, and we report the mean and standard deviation of these 100 values. From the column labelled `Improvement' we see that the KROM algorithm is always faster, and can be more than 10 times faster, than the {\color{blue} trial-and-error method }.

\begin{table}[H]
	\centering
	\footnotesize
	\caption{{\footnotesize{\textbf{Computational Speed Comparison.}\\The simulations have $m$ observations on $n$ random variables with zero mean, identity covariance and a random target Kollo skewness vector with each element randomly chosen from $\mathcal{U}(-1,1)$.  Means and standard deviations (in brackets) of computation times are derived from 100 random vectors per pair $(m,n)$  and the arbitrary values are from $\mathcal{N}(0,0.7)$. The KROM algorithm is compared with the algorithm proposed by { \cite{hanke2017} (HPSW)} using the improvement ratio of mean computation times. The machine used is an Intel Core i5-8250U CPU, 3.41 gigahertz, with 8 gigabyte RAM, running python under Windows.}}}
	\begin{tabular}{ccccccc}
		\hline \\[-9pt]
		\multirow{2}[0]{*}{$m$} & \multicolumn{3}{c}{$n=5$} & \multicolumn{3}{c}{$n=10$} \\
		& \multicolumn{1}{l}{HPSW} & \multicolumn{1}{l}{KROM} & \multicolumn{1}{l}{Improvement} & \multicolumn{1}{l}{HPSW} & \multicolumn{1}{l}{KROM} & \multicolumn{1}{l}{Improvement} \\
		\hline\\[-9pt]
		\multirow{2}[0]{*}{$5n$} & 0.670 & 0.117 & \multirow{2}[0]{*}{5.73} & 2.170 & 0.243 & \multirow{2}[0]{*}{8.93} \\
		& (0.424)& (0.056) &       & (1.326) & (0.091) &  \\
		\hline\\[-9pt]
		\multirow{2}[0]{*}{$10n$} & 1.007 & 0.099 & \multirow{2}[0]{*}{10.17} & 2.280 & 0.262 & \multirow{2}[0]{*}{8.70} \\
		& (0.540) & (0.040) &       & (0.788) & (0.112) &  \\
		\hline\\[-9pt]
		\multirow{2}[0]{*}{$20n$} & 1.425 & 0.163 & \multirow{2}[0]{*}{8.84} & 2.747 & 0.409 & \multirow{2}[0]{*}{6.72} \\
		& (0.465) & (0.107) &       & (2.032) & (0.125) &  \\
		\hline\\[-9pt]
		\multirow{2}[0]{*}{$50n$} & 2.788 & 0.279 & \multirow{2}[0]{*}{9.99} & 3.670 & 0.710 & \multirow{2}[0]{*}{5.17} \\
		& (3.340) & (0.156) &       & (4.691) & (0.200) &  \\
		\hline\\[-9pt]
		\multirow{2}[0]{*}{$100n$} & 3.792 & 0.384 & \multirow{2}[0]{*}{9.88} & 15.288 & 6.529 & \multirow{2}[0]{*}{2.34} \\
		& (3.238) & (0.165) &       & (14.720) & (9.898) &  \\
		\hline
	\end{tabular}%
	\label{tab:addlabel}%
\end{table}%

The  conditions derived in Theorem \ref{theorem HPSW} for the Kollo skewness equations to be solvable in the real numbers depend on the target Kollo skewness $\bm{\tau}_{n}$, the system dimension $n$ and the simulation sample size $m$. There is no guarantee and it may be that, for certain values of these parameters, all bootstraps or draws from a parametric distribution are inadmissible, i.e. there is no real solution to \eqref{eq:kollo:p}  -- \eqref{eq:kollo:sum}.  Then, the KROM algorithm fails because the user has specified parameters that are inconsistent with a solution. How does the failure rate of KROM simulation algorithm change with different $n$ and $m$ and how does it depend on $\sigma$, the standard deviation of the observations drawn via bootstrap or using a parametric distribution? Answers to these questions will help a researcher implement the KROM algorithm in practice. 
		
To provide at least indicative answers,  we perform an experiment by drawing  $z_{ij}$ independently  from $\mathcal{N}(0, \sigma^2)$,  setting the target $\bm{\tau}_{n}=\mathbf{0}_{n}$ for consistency, because we know from Appendix \ref{app: attainable p1} that one should use another distribution, such as the skew normal, the normal inverse Gaussian or the beta distribution, when the target Kollo skewness is not zero.  Table \ref{tab1} reports the results using 10,000 normal simulations, for different pairs $(n, m)$. Because the computational time increases with $m$, we simply report the proportion of these 10,000 draws for which the Kollo skewness equations have no solution. If this is 100\% then the algorithm might even fail completely however many draws from $\mathcal{N}(0, \sigma^2)$ are used. We also report the proportion of draws from a normal distribution for which Lemma 1 holds, i.e. the inequality (17)  is valid.  This is because, for the first column, $j=1$, there is an  additional cubic equation \eqref{eq:system3} which is the most difficult one to solve. For simplicity,  Table \ref{tab1} only reports results for values based on $\mathcal{N}(0,\sigma^2)$ with the same $\sigma$ for each column. To examine the effect of the system dimension, we set $n=5$ and $n=20$.
	
	\begin{table}[H]
		\centering
		\footnotesize
		\caption{\footnotesize{\textbf{Proportion of Draws from $\mathcal{N}(0, \sigma^2)$  Yielding No Real Solutions to Kollo Skewness Equations}.\\ The  values for $z_{ij}$ are independently drawn 10,000 times  from $\mathcal{N}(0, \sigma^2)$ and the table reports the proportion of draws for which the Kollo skewness equations  \eqref{eq:kollo:p} -- \eqref{eq:kollo:sum} have no solution, i.e. the conditions of Theorem 1 do not hold. We also denote by $\alpha_1$ the proportion of draws for which the inequality of Lemma \ref{lemma 1} does not hold. We set $n=5$ and $n=20$, and  $\bm{\tau}_{n}=\mathbf{0}_{n}$ for different $n$ and for different simulation sample sizes  $m=50, 100,500, 1000$. We let $\sigma^2$ vary from 0.6 to 0.9.}}
		\begin{tabular}{cccrrrrrrrr}
			\hline
			\multirow{2}[0]{*}{$n$} & \multirow{2}[0]{*}{$\bm{\tau}_{n}=\mathbf{0}_{n}$ } & \multicolumn{1}{r}{\multirow{2}[0]{*}{$\sigma^2$}} & \multicolumn{2}{c}{$m=50$} & \multicolumn{2}{c}{$m=100$} & \multicolumn{2}{c}{$m=500$} & \multicolumn{2}{c}{$m=1000$} \\
			&       &       & \multicolumn{1}{c}{$\alpha_1$} & \multicolumn{1}{c}{$\alpha$} & \multicolumn{1}{c}{$\alpha_1$} & \multicolumn{1}{c}{$\alpha$} & \multicolumn{1}{c}{$\alpha_1$} & \multicolumn{1}{c}{$\alpha$} & \multicolumn{1}{c}{$\alpha_1$} & \multicolumn{1}{c}{$\alpha$} \\
			\hline
			\multirow{4}[0]{*}{5} & \multirow{4}[0]{*}{$\mathbf{0}_5$} & 0.6   & 0   & 0.62 & 0   & 0.40 & 0   & 0.17 & 0   & 0.23 \\
			&       & 0.7   & 0.31 & 12.34 & 0   & 1.74 & 0   & 0.18 & 0   & 0.29 \\
			&       & 0.8   & 10.83 & 70.80 & 1.25 & 38.67 & 0   & 0.51 & 0   & 0.16 \\
			&       & 0.9   & 49.83 & 99.50 & 34.65 & 98.69 & 0.10 & 53.51 & 0   & 14.68 \\
			\hline
			\multirow{4}[0]{*}{20} & \multirow{4}[0]{*}{$\mathbf{0}_{20}$} & 0.6   & 0.01 & 63.56 & 0   & 28.06 & 0   & 0.63 & 0   & 0.18 \\
			&       & 0.7   & 0.44 & 99.78 & 0.01 & 99.36 & 0   & 1.74 & 0   & 0.15 \\
			&       & 0.8   & 10.86 & 100 & 1.03 & 100 & 0   & 98.73 & 0   & 25.52 \\
			&       & 0.9   & 49.49 & 100 & 34.50 & 100 & 0.06 & 100 & 0   & 100 \\
			\hline
		\end{tabular}%
		\label{tab1}%
	\end{table}%

 For any given $m$ and $n$ in Table \ref{tab1}, the following properties are observed: all draws from any normal distribution yield real roots for the first column, that is, the inequality condition (17) of Lemma \ref{lemma 1} always holds, especially for low values of $\sigma^2$;  this is not affected by the dimension $n$ of the system but (17)  becomes less likely to hold as $\sigma^2$ increases ;  and finally, as either $n$ or $\sigma^2$ increases,  the entire system of Kollo skewness equations also becomes more difficult to solve -- e.g. when $\sigma^2 = 0.9$ it is highly probable that the KROM algorithm will not work for $n=20$ or more. These results provide a numerical answer to the question of how small  $\sigma^2$ should be, given the dimension of the system and the number of simulations in each sub-sample, prior to concatenation. 

 If time is not an issue, the prolonged `trial-and-error' generation of  arbitrary values in the \cite{hanke2017} algorithm may be acceptable. However, in some applications it can be important to perform a very large number of simulations in an extremely short time -- for instance, when pricing financial products in an algorithmic trading strategy. In that case, we may wish to define a prescribed acceptable failure rate $\alpha$, e.g. $\alpha = 0.05$. When might a prescribed failure rate be useful? Much depends whether on the time taken for the algorithm to work, which itself depends on the dimension of the system and the number of simulations in each sub-sample. For instance, using the code that we provide with this paper, running Python under Windows on a  computer with Intel Core i5-8250U CPU, 3.41 GHz, and 8 GB RAM, it takes only $0.0375$ seconds to generate 1000 simulations for a 20-dimensional system, but this is using zero arbitrary values when $\bm{\tau}_{n}=\mathbf{0}_{n}$. It takes $1.956$ seconds to generate  the same number of simulations when drawing from  $\mathcal{N}(0, 0.5)$ and $4.706$ seconds with  $\mathcal{N}(0, 0.8)$. These times increase rapidly with $n$ especially when the target Kollo skewness is larger. 

 In practice, we  utilise sample concatenation  to improve the statistical features of ROM simulations. For instance, consider the extreme case where all arbitrary values for $n$ variables are zero so that a direct application of the { \cite{hanke2017}} algorithm yields only $n+2$ rows with non-zeros values. This results in some strange features especially when the simulation size $m$ is large.  To see this, consider an hypothetical example with $n=5$, $m=100$, a target mean vector $\mathbf{0}_5$, a target covariance matrix $\mathbf{1}_5$ and Kollo skewness vector $\mathbf{1}_n$. Figure \ref{fig:n5all0} plots the resulting sequential graphs and histograms for the corresponding ROM simulations. 

 Clearly, the Kollo skewness structure is captured by the first few elements alone -- the rest of the simulation is zero.  If required,  application of a random permutation matrix in the ROM simulation can change the location of this activity, as discussed in \cite{ledermann2011random}, but such patterns are not appropriate in many fields of application. Nevertheless, {\color{blue} setting arbitrary values to zero does have one advantage -- it reduces the failure rate.} 
\begin{figure}[H]
	\caption{\footnotesize{\textbf{{Samples Generated by the \cite{hanke2017} Algorithm with Zero Arbitrary Values}.}\\
	Setting $m=100$, $n=5$, $\bm{\tau}=\bm{1}_{5}$ 	and depicting  sequential figures (above) and histograms (below).}}
	\label{fig:n5all0} 
	\includegraphics[width=17.5cm,height=5cm,trim={1cm 0.5cm 0.3cm 0.5cm clip}]{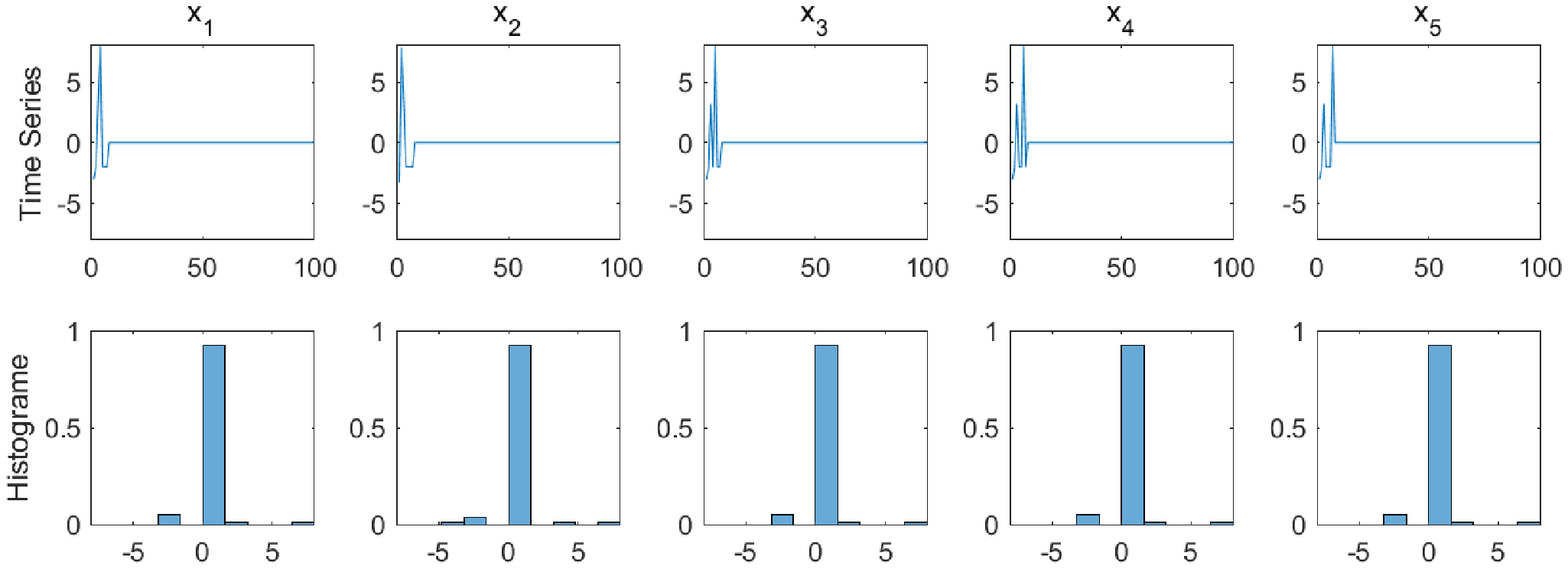} 
\end{figure}

In KROM simulation applications, such as those in finance, we advocate using a bootstrap rather than a parametric distribution to populate the undetermined values of the Kollo skewness equations, for reason which shall presently become clear.  However,  because  we always need to apply a variance reduction in equation \eqref{eq:adjust} to ensure admissibility of the values selected, excessively high kurtosis could be problem for some applications, and this is where we advocate the use of the concatenation option. To illustrate just how high the kurtosis can be, and how concatenation reduces this, compare choosing arbitrary values from $\mathcal{N}(0,0.3)$ with those from $\mathcal{N}(0,0.6)$, otherwise setting the same parameters as those for Figure \ref{fig:n5all0}. 
The first column of Figure \ref{fig:n5_b1} exhibits the corresponding histograms of KROM simulations. The results for $\sigma^2=0.3$ in blue show many extreme values and a much greater kurtosis than the result for $\sigma^2=0.6$ in orange.  Next, we apply concatenation, splitting $m$ into $N$ blocks, and applying the KROM algorithm separately for each block. 
The histograms in the middle and right-hand columns of Figure \ref{fig:n5_b1} depict the histograms obtained using sub-sample sizes 10 and 20. The figures also display  the corresponding marginal kurtosis. Without concatenation, marginal kurtosis is extremely high especially using $\mathcal{N}(0,0.3)$ (blue).  In each case and for each variable, concatenation reduces this, especially when the sub-sample size in the concatenation is small.

\begin{figure}[H]
	\caption{\footnotesize{\textbf{{The Effect of Concatenation in KROM Simulation}} \\The left-hand column shows the  results generated by the basic KROM algorithm with $n=5$, $m=100$, $\bm{\tau}=\bm{1}_{5}$ and the other columns show the effect of the concatenation extension. With concatenation, the  length of the sub-samples are $l=20$ (middle column) and $l=10$ (right-hand column). The rows correspond to variables $x_{1},\ldots,x_{5}$. The blue and orange shows the sample using  two normal distributions, $\mathcal{N}(0,0.3)$ (blue) and $\mathcal{N}(0,0.6)$ (orange).}} \label{fig:n5_b1} 
	\includegraphics[width=18cm,height=9cm,trim={2cm 1cm 2.5cm 0.2cm clip}]{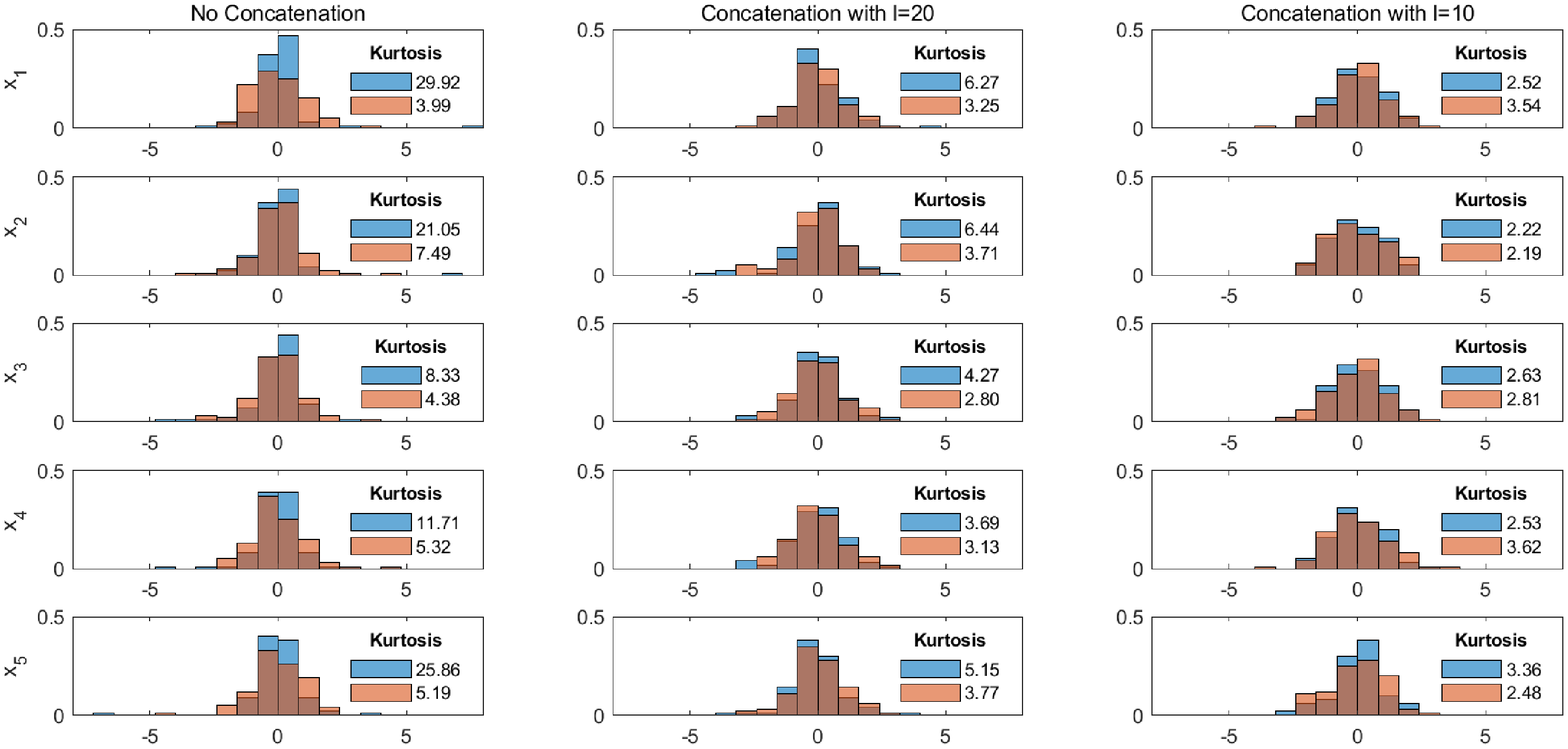}
	
\end{figure}

\section{Empirical Study} \label{sec: real data}
Here we illustrate the application of KROM simulation to two different systems of financial returns, based on real data. Section \ref{sec: real data first appliation}  examines the empirical time series and cross-sectional properties of the Kollo skewness vector and its rotation. Section \ref{sec: real data second appliation} demonstrates the usefulness of KROM simulation using bootstrapping  compared with drawing simulations from a parametric distribution. Section \ref{sec:RMSE}  investigates the sampling errors for mean, covariance and Kollo skewness when a standard bootstrap is applied. The size and significance of these errors validates one of the main advantages of our algorithm which, of course, always has zero sampling error by design.  

\subsection{Kollo Skewness Properties of Financial Data} \label{sec: real data first appliation}
We consider two systems of returns depicted in Figure \ref{fig:orignal TS}.  On the left are hourly returns on US dollar prices of three cryptocurrencies: bitcoin (BTC), ether (ETH), and litecoin (LTC) from 1 January 2017 to 14 August 2020, over 27,000 data points in each series. The other data set is for eleven S\&P 500 industry sector indices consisting of companies in the same or related industries. Here we employ daily returns from 16 October 2001 to 18 September 2020 and there are over 4,500 data points in each series. The right panel of Figure \ref{fig:orignal TS} only depicts three of the eleven indices, viz. Energy, Finance and Real Estate.
\begin{figure}[H]
		\caption{\textbf{\footnotesize{}Time Series of Financial Returns.} \\
		{\footnotesize{}Panel (a) exhibits the hourly returns on three cryptocurrencies: BTC, ETH and LTC and panel (b) exhibits the daily returns for three S\&P 500 sector indices, viz. Energy, Finance and Real Estate}}
	\label{fig:orignal TS} 
	\centering
	\subfigure[Cryptocurrency]{
		\includegraphics[width=6.5cm,height=4.5cm,trim={1cm 0.5cm 1cm 0cm clip}]{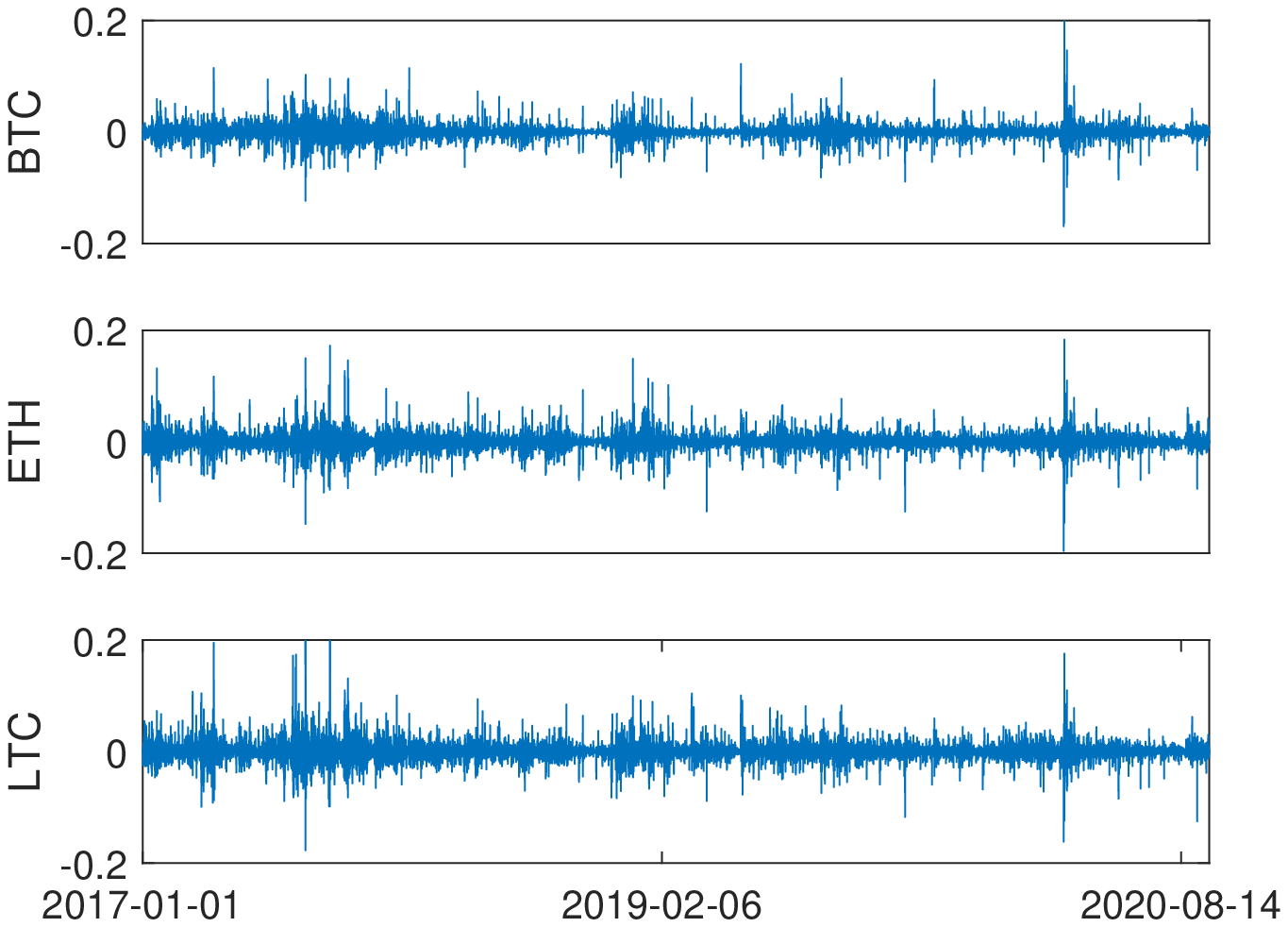} 
	}\qquad
	\subfigure[Sector]{
		\includegraphics[width=6.5cm,height=4.5cm,trim={1cm 0.5cm 1cm 0cm clip}]{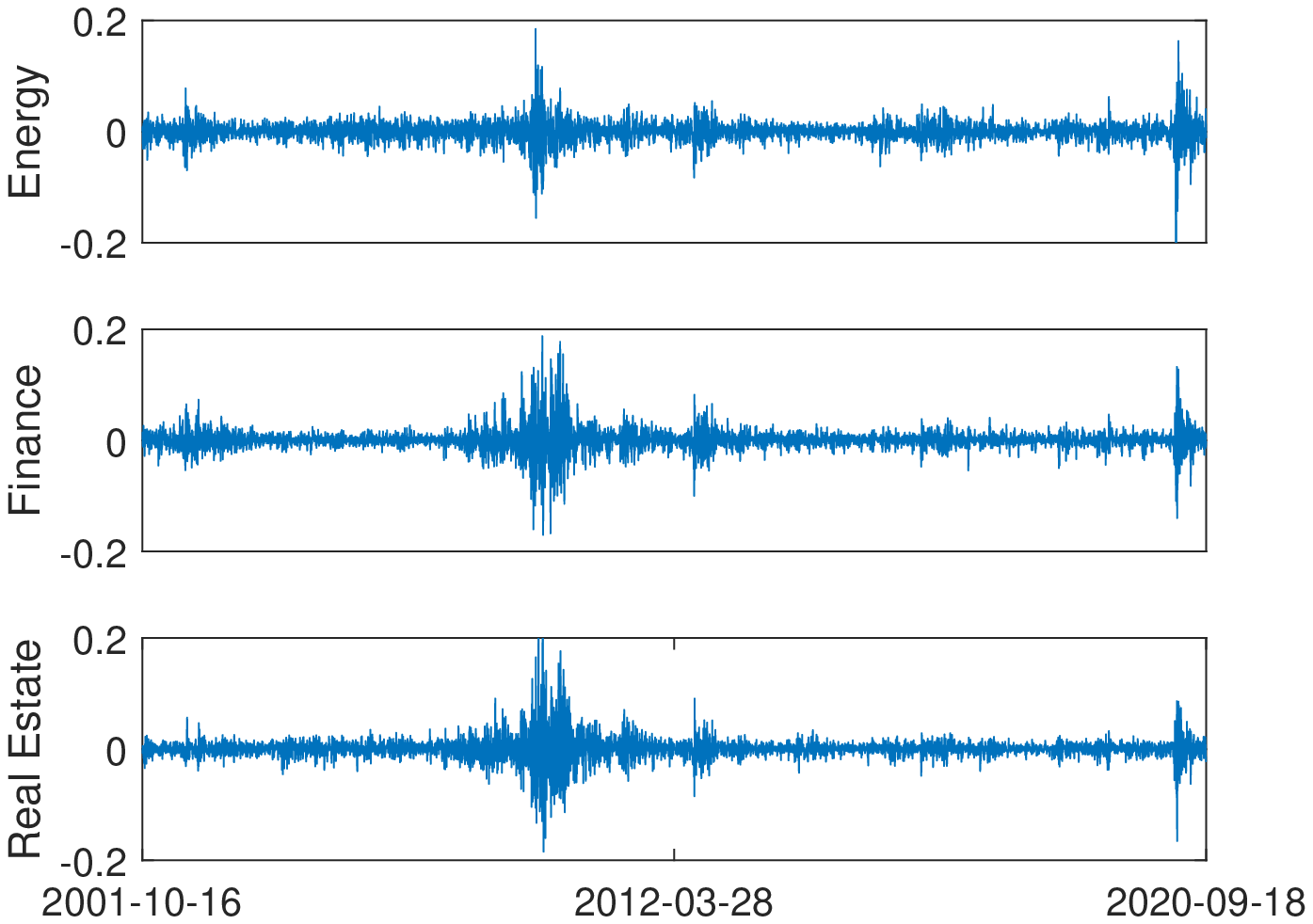}
	}

\end{figure}
First we calculate the Kollo skewness vector using a rolling window to illustrate the range of values that could be expected using real data. 
For the hourly cryptocurrency data, setting the length of rolling window 720 points (i.e. around one month) all elements in the Kollo skewness vector lie in the range $[-10,10]$ with most sampling variability in $\tau_1$, as shown in panel (a) of Figure \ref{fig: crypto kollo skewness}. In particular, the cryptocurrency market-wide price crash on 12 January 2020 (Black Thursday) is evident from the dip in Kollo skewness down to almost $-10$.  On the right, panel (b) of Figure \ref{fig: crypto kollo skewness} depicts the rotated Kollo skewness using the matrix  $\mathbf{\Omega}_{n}$ specified in Section 3.2. Here we note that $p_1, p_2$ and $p_3$ almost always lie in the range $[-3,3]$.  
\begin{figure}[H]
		\caption{\footnotesize{\textbf{Kollo Skewness and its Rotation for the Cryptocurrency Hourly Returns.} \\ 
			The Kollo skewness and its rotation are computed  on a rolling window of 720 hourly returns.}}
	\label{fig: crypto kollo skewness}
	\centering
	\subfigure[Kollo skewness]{
		\includegraphics[width=6.5cm,height=4.5cm,trim={1cm 0.5cm 1cm 0cm clip}]{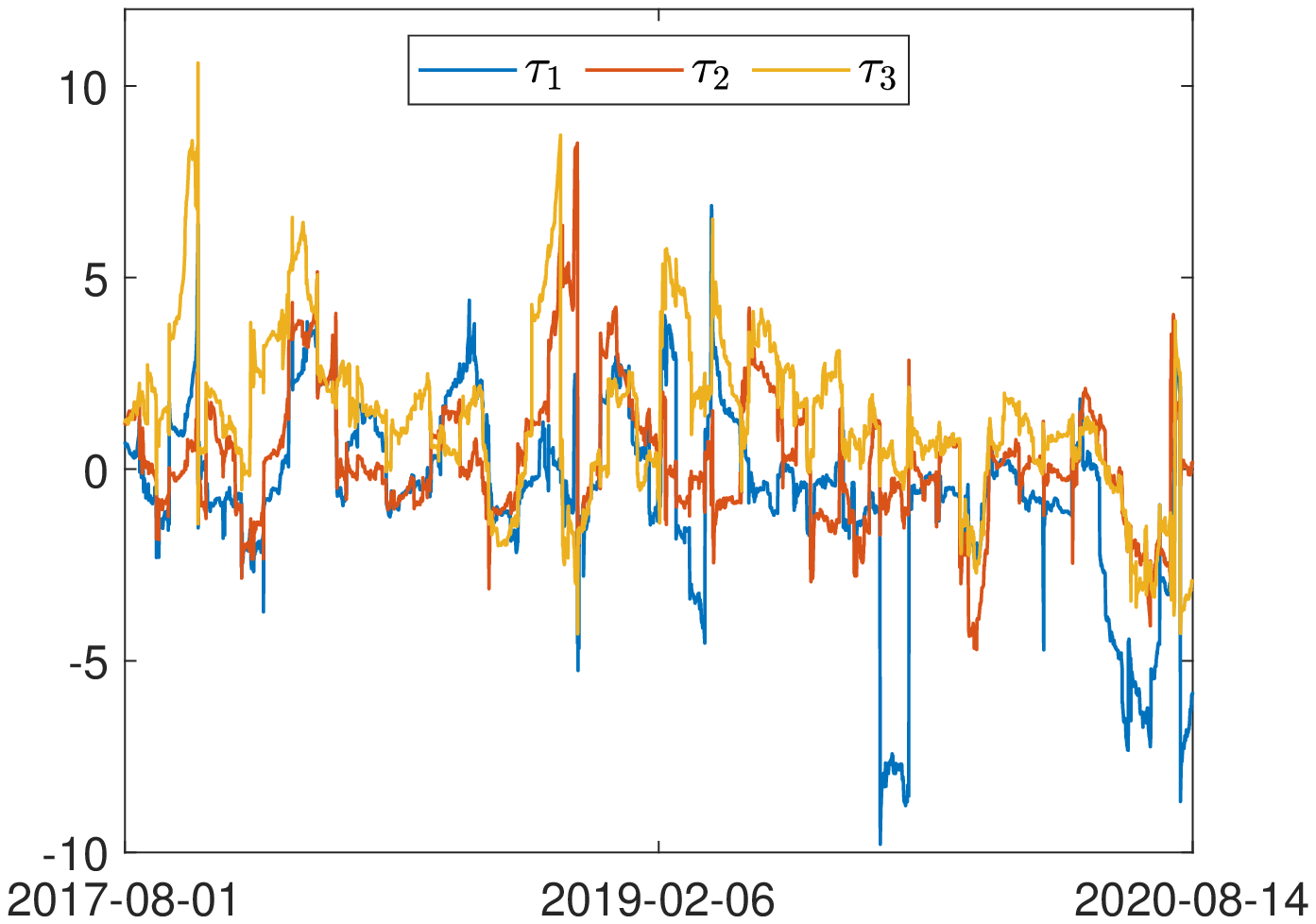}
	}\qquad
	\subfigure[Rotated Kollo skewness]{
		\includegraphics[width=6.5cm,height=4.5cm,trim={1cm 0.5cm 1cm 0cm clip}]{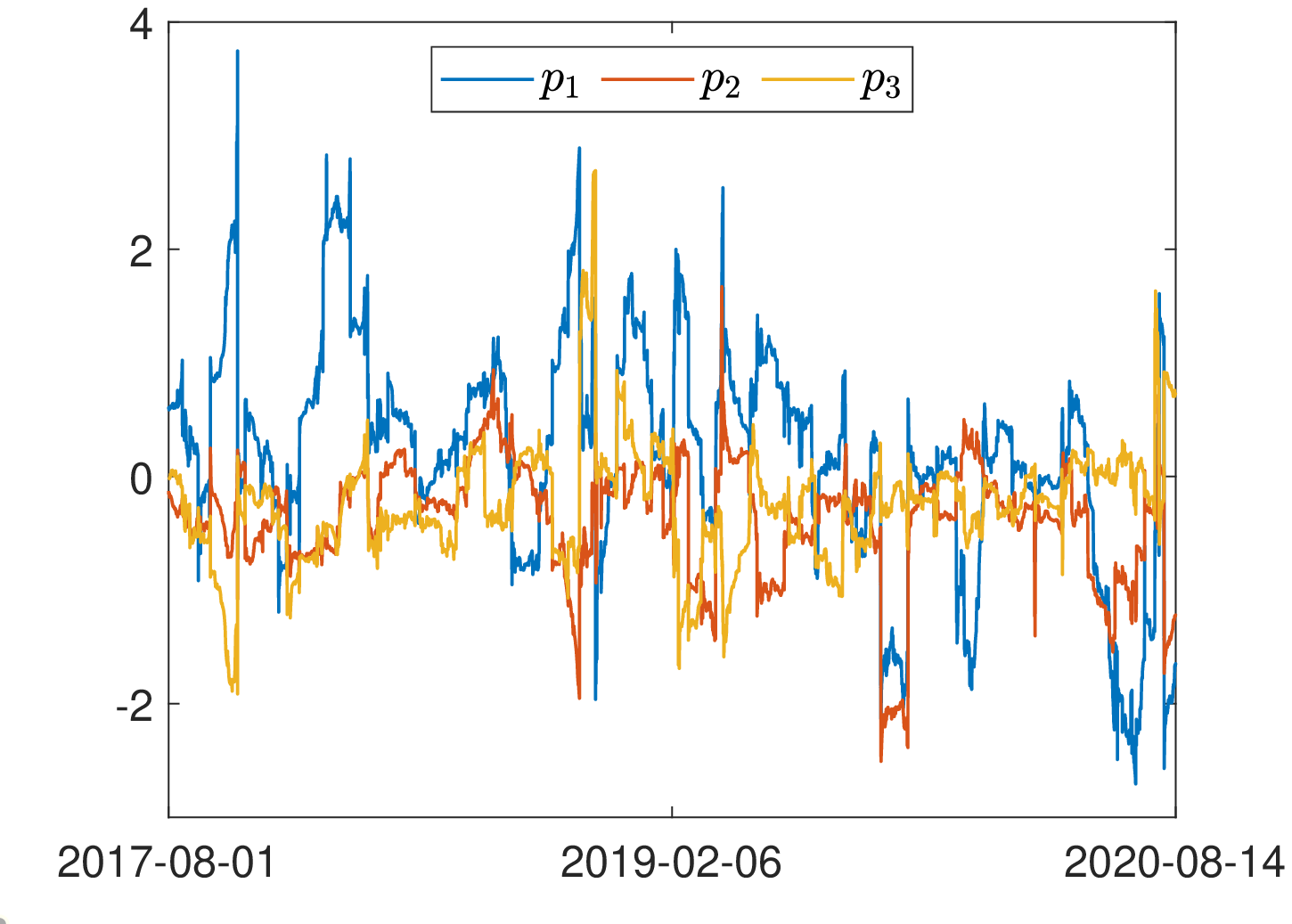}
	}

\end{figure}
Figure \ref{fig: sector kollo skewness} illustrates the results of similar computations for the S\&P 500 sector data, this time using a one-year rolling window containing 250 points. To keep consistency with Figure \ref{fig:orignal TS} we only depict $\tau_1, \tau_7$ and $\tau_{11}$, i.e. the Energy, Finance and Real Estate sector results, but the 11-dimensional Kollo skewness and its rotated value are in a similar range for the other sectors.
 The Kollo skewness takes a very similar range to the cryptocurrency data except during the sector-wide crash in the S\&P 500 index in March 2020. For example the low values of $\tau_1$  during the first half of 2018 were caused by an upward surge in oil prices.  Similarly, the upward spike in $\tau_7$  on 16 March 2020 resulted from an extreme return in the Finance sector index, of around $-14\%$. Apart from these extreme values, we note that most elements in the rotated Kollo skewness lie in the range $[-1,1]$.

\begin{figure}[H]
	\centering
		\caption{\footnotesize{\textbf{The Energy, Finance and Real Estate elements of Kollo Skewness and its Rotation } \\ 
			Based the daily returns of eleven S\&P 500 sector indices, we depict only three. The Kollo skewness and its rotation are computed on a rolling window 250 daily returns.}}
	\label{fig: sector kollo skewness}
	\subfigure[Kollo skewness]{
		\includegraphics[width=6.5cm,height=4.5cm,trim={1cm 0.5cm 1cm 0cm clip}]{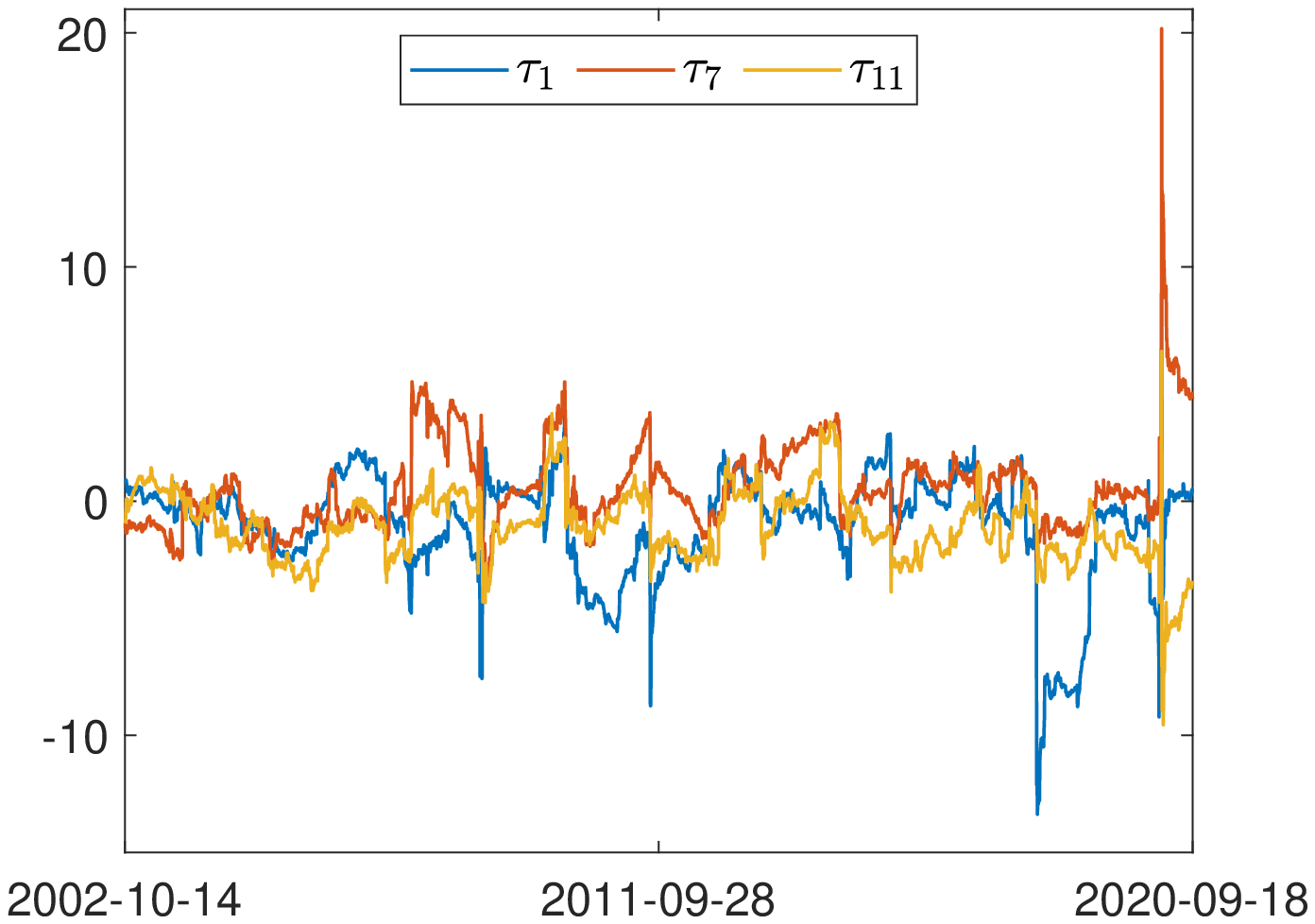}
	}\qquad
	\subfigure[Rotated Kollo skewness]{
		\includegraphics[width=6.5cm,height=4.5cm,trim={1cm 0.5cm 1cm 0cm clip}]{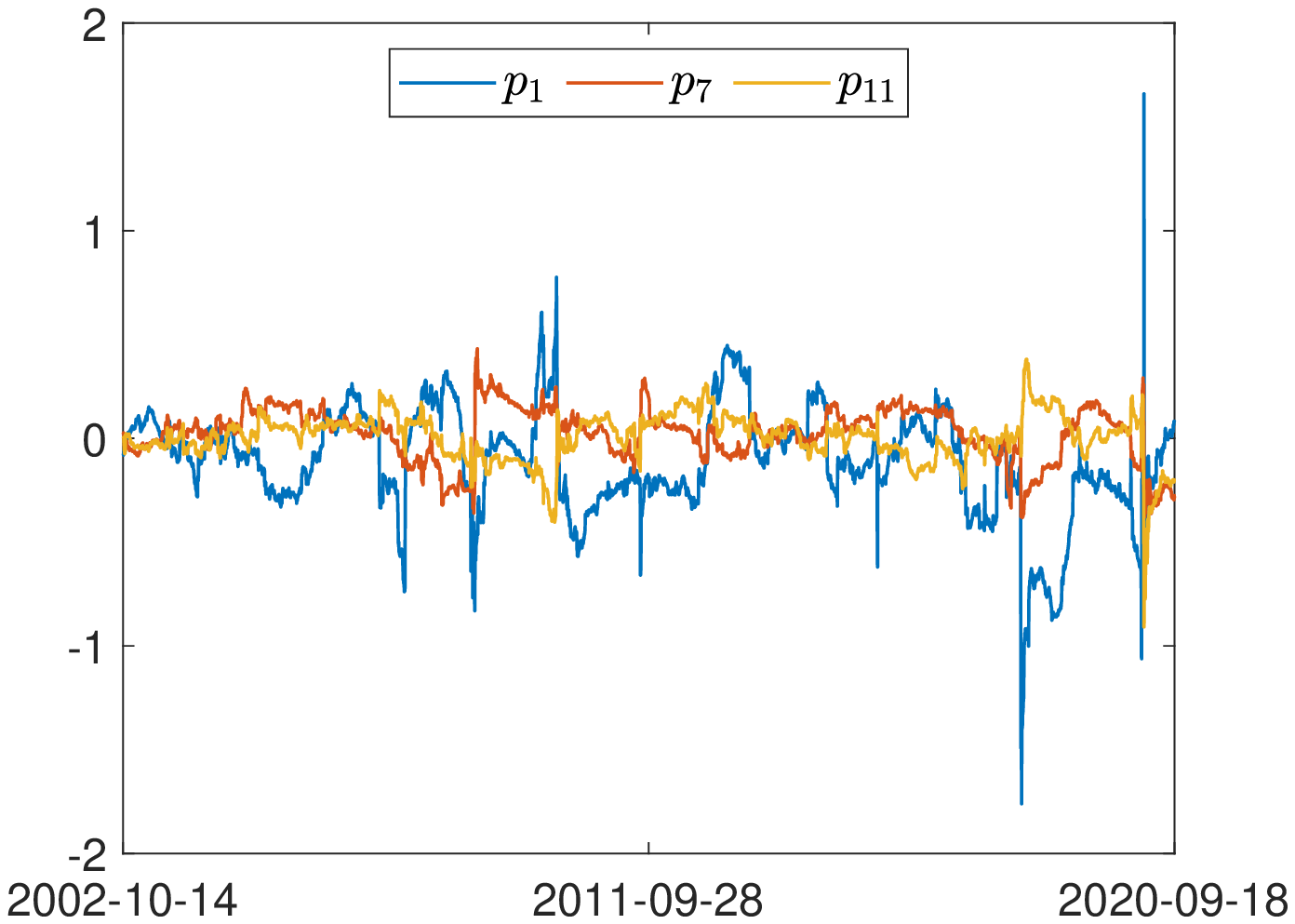}
	}

\end{figure}

\subsection{Applications of KROM Simulation} \label{sec: real data second appliation}

Next we apply the  KROM algorithm to simulate the three-dimensional cryptocurrency hourly returns, based on some historical values for mean, covariance and Kollo skewness, i.e. certain empirical quantiles derived from the sample exhibited in Figure \ref{fig: crypto kollo skewness}. 
We design the experiment with extreme values and a median value by selecting target samples with corresponding $1\%$, $50\%$ and $99\%$ quantiles of $p_1$.  We focus on the first element of  $\mathbf{p}_n$ because  the most strict conditions are the cubic restriction on the first column of $\mathbf{S}_{mn}$, i.e. equation \eqref{eq:kollo:p}.  That is, the $p_1$ time series exhibited in Figure \ref{fig: crypto kollo skewness} panel (b) is used to generate a histogram and the $q$-quantile is identified from this, for $q = 1\%, 50\%$ and $99\%$. For each quantile of $p_1$,  the corresponding values for $p_2$ and $p_3$ are those observed for the same time period, based on same 720 observations for each variable. These $\mathbf{p}_3$ vectors are named Case 1, Case 2, and Case 3 in the left-hand column of Table \ref{tab:failure rate for real}. 

To demonstrate that the KROM algorithm works best when we use a bootstrap from a real data sample to solve the Kollo skewness equations, we compare this approach  to using some parametric  distributions, viz. the normal distribution (N), the skew-normal  (SN), the beta  (Beta), and the normal inverse Gaussian distribution (NIG). In each case we set  $\sigma ^2$ to either 0.7, 0.8 and 0.9 and fix the other parameters by fitting to each sample of 720 observations for each variable. This way, the distributions are parametrised to reflect the properties of each marginal empirical distribution.  Table \ref{tab:failure rate for real} reports the results, in terms of the proportion of 10,000 trials for which we cannot find a solution to the Kollo skewness equations. The table reports this for different values of $\bm{\tau}_3$, $m$ and $\sigma$ as defined in \eqref{eq:adjust}. The parameters of each distribution are fixed throughout.

\begin{table}[htbp]
	\centering
	\footnotesize
	\caption{\footnotesize{\textbf{KROM Simulation: Bootstrapping from Real Data vs Drawing from Parametric Distributions}\\ We report the proportion among 10,000 trails where the Kollo Skewness equations cannot be solved. That is, the  necessary and sufficient conditions in Theorem \ref{theorem HPSW} do not hold. The  three target samples with different Kollo skewness are obtained using historical data and these are rotated using the matrix defined in Section\ref{sec:ext}. We apply KROM to simulate samples of size $m=500$ and $m=1000$ and compare the  different ways of generating values for the Kollo skewness equations: bootstrapping from historical data, assuming a normal distribution (Normal), a skewed-Normal distribution (SN), a Beta distribution (Beta), and a Normal-Inverse Gaussian distribution (NIG), each computed for different values of $\sigma$ in \eqref{eq:adjust}}.}
	\begin{tabular}{ccccrrrrrrrrrr}
		\hline
		& &       &       & \multicolumn{2}{c}{Bootstrapping} & \multicolumn{2}{c}{Normal} & \multicolumn{2}{c}{SN} & \multicolumn{2}{c}{Beta} & \multicolumn{2}{c}{NIG} \\
		& $\bm{\tau}_3$ & $\bm{\tilde{\tau}}_3$  & $\sigma^2$ & 500   & 1000  & 500   & 1000  & 500   & 1000  & 500   & 1000  & 500   & 1000 \\
		\hline\\[-3pt]
		\multirow{3}[0]{*}{1}&\multirow{3}[0]{*}{$\begin{bmatrix}-6.44\\ -2.22\\ -2.97 \end{bmatrix}$} & \multirow{3}[0]{*}{$\begin{bmatrix}-2.24\\	-1.05\\	\ \ 0.18\end{bmatrix}$} & 0.7   & 35.62 & 7.56 & {100} & {97.93} & 96.83 & {0.01} & {0.08} & 0.18 & 36.79 & 4.97 \\
		&  &     & 0.8   & {68.23} & 36.80 & {100} & {100} & {100} & {100} & 95.43 & {1.52} & 79.41 & 44.90 \\
		&   &    & 0.9   & {83.67} & {75.13} & {100} & {100} & {100} & {100} & 99.98 & 99.91 & 90.63 & 84.63 \\[3pt]
		\hline\\[-3pt]
		\multirow{3}[0]{*}{2}&\multirow{3}[0]{*}{$\begin{bmatrix}-1.09\\ \ \ 1.25\\ \ \ 1.08\end{bmatrix}$} & \multirow{3}[0]{*}{$\begin{bmatrix}\ \ 0.24 \\-0.61 \\ \ \ 0.04\end{bmatrix}$} & 0.7   & 34.99 & 31.97 & {0.01} & {0.02} & 1.85 & 0.53 & 0.94 & 0.27 & 38.77 & 34.60 \\
		&   &    & 0.8   & 35.31 & 31.18 & {0.65} & {0.02} & 2.02 & 0.53 & 1.32 & 0.41 & 39.35 & 35.28 \\
		&    &   & 0.9   & {46.84} & {46.67} & {100} & {97.86} & 76.44 & 71.77 & 79.28 & 73.18 & 57.59 & 54.07 \\ [3pt]
		\hline \\[-3pt]
		\multirow{3}[0]{*}{3}& \multirow{3}[0]{*}{$\begin{bmatrix} \ \ 2.47\\ \ \ 3.82\\ \ \  5.90  \end{bmatrix}$} & \multirow{3}[0]{*}{$\begin{bmatrix} \ \ 2.35\\-0.65\\ 	-0.49\end{bmatrix}$} & 0.7   & 44.52 & 9.44 & {100} & {100} & 99.72 & {0.03} & {0.14} & 0.39 & 37.85 & 4.75 \\
		&    &   & 0.8   & 57.12 & 33.61 & {100} & {100} & {100} & {100} & {51.79} & {1.34} & 67.74 & 38.98 \\
		&     &  & 0.9   & {74.57} & {69.05} & {100} & {100} & {100} & {100} & 99.51 & 99.32 & 87.65 & 82.21 \\[3pt]
		\hline
	\end{tabular}\label{tab:failure rate for real}
\end{table}%

As anticipated from our numerical results in Section \ref{numerical results} the smallest failure rates occur for larger $m$ and smaller $\sigma^2$. In addition, we find that extremely skewed data (Case 1 and Case 3) yields a substantially higher failure rate than  moderately skewed data (Case 2). Of all methods for generating arbitrary values the historical bootstrap has the most stable performance and much the lowest failure rate except when $m$ is large and $\sigma^2$ is small.  Comparing the failure rate for the parametric distributions, as expected we find the normal distribution only works well around the median range, i.e. when the rotated Kollo skewness is about $\mathbf{0}$. 
Clearly, Table \ref{tab:failure rate for real} suggests that bootstrapping from a historical sample (when possible) produces the lowest failure rates in general. This is not surprising, and in fact, in some sense tautological: the simulations generated using  historical data should retain  most of the properties of that data. 

{\color{blue} 
There are many potential applications of  KROM simulation to financial risk analysis. 
Standard parametric or bootstrap simulation methods are often used to measure quantile-based risk metrics such as value-at-risk \citep{Duffie1997}. However, the considerable simulation error inherent in these approaches can distort results. In an attempt to reduce this risk, a very large number of time-consuming simulations are commonly used, especially when measuring risk in large systems \citep{Jorion1996}. However, the focus is usually on reducing errors in the simulated covariance matrix only, ignoring errors in higher moments. In the next sub-section we demonstrate that bootstrap simulation errors in Kollo skewness are even greater than those for the covariance matrix, so both parametric and bootstrap versions of KROM simulation methods have great advantages for quantile risk measurement. KROM methods are able to simultaneously target both Kollo skewness and covariance matrices  exactly, this way generating precise risk measures at each and every simulation. 

Another application  is to stress testing portfolios. Recall that the bootstrap version of the KROM algorithm first standardizes all returns to have zero mean and unit variance, and then applies a variance shrinkage parameter $\sigma^2$ in \eqref{eq:adjust}. Because lower values of $\sigma^2$ are more efficient, Kollo skewness targeting is well suited to Cases 1 and 3, i.e. where the target vector corresponds to stressful periods represented by the 1\% or 99\% quantiles of the historical Kollo skewness distributions. Then, the other observations that are bootstrapped would correspond to more normal market conditions. This way, based on a rolling window over each (possibly concatenated) KROM simulation, one could construct a time-varying series of risk metrics. Kollo skewness is flexible because it only requires target moments so it is convenient for all forms of scenario analyses, not only stress testing.  
	
	Clearly, KROM simulation can be useful for computing minimum required capital for market risk in banks. Under the Basel II Accord, this capital depends on quantile risk metrics linked to both normal and stressed market conditions. It also has potential applications to portfolio optimization when the allocation problem aims to maximise expected profits whilst minimising some risk metric. For example, \cite{ledermann2010thesis} analyses the portfolio allocation problem with a mean plus quantile-based selection criterion. Although that application uses standard ROM simulation, one could also use KROM simulation to incorporate the co-skewness between different assets in the optimization objective.}

\subsection{Sampling Errors from Standard Bootstrapping}\label{sec:RMSE}

We emphasize that using KROM simulation with bootstrapping from real data is fundamentally different from the standard bootstrap approach. When targeting the mean, covariance matrix and Kollo skewness vector, our method has zero sampling error in each and every KROM simulation. However, a standard bootstrap approach would inevitably suffer from sampling errors. To illustrate the advantage of the KROM approach we now examine these sampling errors using a root mean square  error (RMSE). For any given $n$, if $\mu_i, V_{ij}, \tau_i$ are the elements of the target mean $\bm{\mu}_n$, covariance $\mathbf{V}_{n}$ and Kollo skewness $\bm{\tau}_n$ respectively and $\hat{\mu}_i, \hat{V}_{ij},\hat{\tau}_i$ are the corresponding elements of the sample moments obtained by standard bootstrapping, the RMSE  are defined, respectively, by:

\begin{equation}\label{eqn:se}
\sqrt{\sum_{i}^{n}(\mu_{i}-\hat{\mu}_i)^2},\,
\sqrt{\sum_{i,j}^{n}(V_{ij}-\hat{V}_ij)^2}, \,
\sqrt{\sum_{i}^{n}(\tau_{i}-\hat{\tau}_i)^2}. 
\end{equation}
Table \ref{tab: sample error for bootstrapping} reports sampling errors for the mean vector, covariance matrix and Kollo skewness vector when applying the statistical bootstrapping to the cryptocurrency historical returns. Again we set rotated Kollo skewness to Cases 1, 2 and 3 as described above for Table \ref{tab:failure rate for real}. Note that  720 hourly returns were used to generate  each $q$-quantile. It is these returns that we use for each case of the statistical bootstrap. This way, we know the multivariate sample moments and we use 720 data points to bootstrap  samples of a fixed size $m$,  with replacement. We consider sample  sizes  $m=100,\, 500$ and $1000$. In order to preserve correlation, we randomly select entire rows of the target sample. 
We do 10,000 bootstraps so the RMSE above are calculated 10,000 times, and this way we obtain both a mean and a standard error of each quantity in \eqref{eqn:se}.  
The right-hand panel of Table \ref{tab: sample error for bootstrapping} reports the ratio of the mean to the standard error, using the central limit theorem to derive critical values for significance. The results show that sampling errors are very significant for both the mean and Kollo skewness, especially when $m$ is small. 

\begin{table}[H]
	\centering
	\footnotesize
	\caption{\footnotesize{\textbf{RMSE of Multivariate Moments in Standard Statistical Bootstrapping.}\\
 	The first three columns illustrate the first three moments for three sets of samples each of size 720 and the fourth column states the size of each bootstrap. The mean and standard error (in parentheses, below) derived using \eqref{eqn:se} are in columns 5 --7,  computed using 10,000 bootstraps for each. The last three columns are the ratio of the average RMSE to its standard error and we use *, ** or *** to denote significantly greater than 0 at the 90\%, 95\% or 99\% confidence level, respectively.}}
 \label{tab: sample error for bootstrapping}%
	    \begin{tabular}{ccccccclll}
	    	\hline	
	    	$\bm{\mu}_{3}$& $\mathbf{V}_{3}$ & $\bm{\tau}_{3}$ & $m$ & \multicolumn{3}{c}{RMSE} & \multicolumn{3}{c}{Ratio} \\
	    	$(10^{-5})$& $(10^{-5})$  &   &    & $\bm{\mu}_{3}(10^{-5})$ & $\mathbf{V}_{3}(10^{-5})$ & $\bm{\tau}_{3}$ & $\bm{\mu}_{3}$    & $\mathbf{V}_{3}$   & $\bm{\tau}_{3}$ \\
	    	\hline
	    	\multirow{6}[0]{*}{$\begin{bmatrix}\ \ 6.91\\20.53\\ \ \, 0.06 \end{bmatrix}$} & \multirow{6}[0]{*}{$\begin{bmatrix} 2.93&    2.99  &2.74 \\ 2.99& 4.40&3.49\\2.74&3.49&3.65\\ \end{bmatrix}$}
	   		&\multirow{6}[0]{*}{$\begin{bmatrix}-6.44\\ -2.22\\ -2.97 \end{bmatrix}$} & \multirow{2}[0]{*}{100} & 8.739 & 5.017 & 7.409 & \multirow{2}[0]{*}{1.500*} & \multirow{2}[0]{*}{1.302*} & \multirow{2}[0]{*}{2.367***} \\
	    	&       &       &       & (5.827) & (3.855) & (3.131) &       &       &  \\
	    	&       &       & \multirow{2}[0]{*}{500} & 3.939 & 2.303 & 4.563 & \multirow{2}[0]{*}{1.608*} & \multirow{2}[0]{*}{1.418*} & \multirow{2}[0]{*}{1.587*} \\
	    	&       &       &       & (2.449) & (1.624) & (2.875) &       &       &  \\
	    	&       &       & \multirow{2}[0]{*}{1000} & 2.823 & 1.607 & 3.141 & \multirow{2}[0]{*}{1.617*} & \multirow{2}[0]{*}{1.427*} & \multirow{2}[0]{*}{1.516*} \\
	    	&       &       &       & (1.746) & (1.126) & (2.072) &       &       &  \\
	    	\hline
	    	\multirow{6}[0]{*}{$\begin{bmatrix}-4.11\\ \, 14.12 \\-2.31 \end{bmatrix}$} & \multirow{6}[0]{*}{$\begin{bmatrix}2.61 & 2.59&2.52\\	2.59& 4.74& 3.96\\ 2.52 & 3.96 & 5.51\\ \end{bmatrix}$} &\multirow{6}[0]{*}{$\begin{bmatrix}-1.09\\ \ \ 1.25\\ \ \ 1.08\end{bmatrix}$} & \multirow{2}[0]{*}{100} & 9.834 & 3.927 & 3.731 & \multirow{2}[0]{*}{1.691**} & \multirow{2}[0]{*}{1.314*} & \multirow{2}[0]{*}{2.150**} \\
	    	&       &       &       & (5.817) & (2.989) & (1.735) &       &       &  \\
	    	&       &       & \multirow{2}[0]{*}{500} & 4.379 & 1.830 & 2.652 & \multirow{2}[0]{*}{1.747**} & \multirow{2}[0]{*}{1.534*} & \multirow{2}[0]{*}{2.232**} \\
	    	&       &       &       & (2.506) & (1.193) & (1.188) &       &       &  \\
	    	&       &       & \multirow{2}[0]{*}{1000} & 3.111 & 1.301 & 1.892 & \multirow{2}[0]{*}{1.742**} & \multirow{2}[0]{*}{1.514*} & \multirow{2}[0]{*}{1.841**} \\
	    	&       &       &       & (1.786) & (0.859) & (1.028) &       &       &  \\
	    	\hline
	    	\multirow{6}[0]{*}{$\begin{bmatrix}-21.09 \\ \ \ 86.00\\-21.58 \end{bmatrix}$} &\multirow{6}[0]{*}{$\begin{bmatrix} 26.56& 25.41& 27.88\\ 25.41 &43.41&  32.83\\  27.88  &32.83& 41.75\end{bmatrix}$}
	    	&\multirow{6}[0]{*}{$\begin{bmatrix} \ \ 2.47\\ \ \ 3.82\\ \ \  5.90  \end{bmatrix}$} &
	    	\multirow{2}[0]{*}{100} & 28.530 & 28.750 & 7.314 & \multirow{2}[0]{*}{1.669**} & \multirow{2}[0]{*}{1.213} & \multirow{2}[0]{*}{2.435***} \\
	    	&       &       &       & (17.091) & (23.705) & (3.003) &       &       &  \\
	    	&       &       & \multirow{2}[0]{*}{500} & 12.836 & 13.782 & 4.790 & \multirow{2}[0]{*}{1.681**} & \multirow{2}[0]{*}{1.469*} & \multirow{2}[0]{*}{2.026**} \\
	    	&       &       &       &( 7.634) & (9.385) & (2.365) &       &       &  \\
	    	&       &       & \multirow{2}[0]{*}{1000} & 9.114 & 9.768 & 3.270 & \multirow{2}[0]{*}{1.674**} & \multirow{2}[0]{*}{1.456*} & \multirow{2}[0]{*}{1.554*} \\
	    	&       &       &       & (5.445) & (6.710) & (2.104) &       &       &  \\
	    	\hline
	    \end{tabular}%
\end{table}%

\section{Summary and Conclusions} \label{sec:conclusion} 
The existing ROM simulation algorithm that can produce multivariate sample having exact Kollo skewness  proceeds by solving a system of simultaneous non-linear equations numerically,  in an iterative manner.   Because the system is  under-determined, the simulations are generated from a very large number of so-called `arbitrary' values.    However, the trial-and-error approach to finding arbitrary values that the authors advocate may fail, and even when it works the algorithm is relatively slow. Furthermore, it produces simulations with very long periods of inactivity and excessively high kurtosis, and so it is difficult to find any real-world application except, perhaps, to specific problems in seismic wave dynamics or cardiology. 

This paper introduces a new methodology for ROM simulation to target Kollo skewness, which we call KROM simulation {\color{blue}   which proceeds in three stages.} The first  replaces the slow, trial-and-error approach with a rapid check of new, necessary and sufficient  conditions for the Kollo skewness equations to have real solutions. The second stage utilizes a bootstrap, if possible and otherwise draws simulations from a parametric distribution, in order that that the simulations exhibit realistic sample characteristics. An optional third stage applies sample concatenation to reduce marginal kurtosis, if desired. The algorithm utilizes and is supported by several new theoretical results, which are proved in the appendices. We also present some detailed numerical analysis on failure rates as well as an empirical study that applies the KROM algorithm to real financial data. By examining how the failure rate depends on the dimension of the system, the number of simulations, the number of sub-samples, and the characteristics of the empirical or parametric distributions, we are able to guide the researcher how best to operate KROM simulation by setting sample sizes, and by selecting distributions and choosing their parameter values. To the best of our knowledge our empirical study is the first application of Kollo skewness to financial data. After describing the sample characteristics of Kollo skewness it focuses on the features of KROM simulation that  appeal to risk management and portfolio optimization applications. We outline the design of such a methodology  but its actual implementation  in practice is a subject for later research.

\begin{spacing}{0.8}
\bibliographystyle{apalike}
\bibliography{rom}
\end{spacing}

\appendix


\section*{Appendix}
\addcontentsline{toc}{section}{Theoretical Appendices} 
\global\long\def\thesubsection{\Alph{subsection}}%
\global\long\def\thesubsubsection{\Alph{subsection}\arabic{subsubsection}}%
\global\long\def\theequation{A.\arabic{equation}}%
\setcounter{equation}{0} \setcounter{subsection}{0}

\subsection{Proofs}
\label{app: theoretical proofs} 
\begin{proof}[Proof of Lemma 1]
	For convenience, let us re-write the system \eqref{eq:system1}--\eqref{eq:system3}
	as follows: 
	\begin{align*}
	s_{11}+s_{21}+s_{31} & =-\sum_{i=4}^{m}w_{i1}=:a,\\
	s_{11}^{2}+s_{21}^{2}+s_{31}^{2} & =m-\sum_{i=4}^{m}w_{i1}^{2}=:b,\\
	s_{11}^{3}+s_{21}^{3}+s_{31}^{3} & =mp_{1}-\sum_{i=4}^{m}w_{i1}^{3}=:c.
	\end{align*}
	Let $x_{1}=s_{11}-\frac{a}{3}$, $x_{2}=s_{21}-\frac{a}{3}$ and $x_{3}=s_{31}-\frac{a}{3}$.
	Then, the above equations reduce to: 
	\begin{align}
	x_{1}+x_{2}+x_{3} & =0,\label{eq:c1_1}\\
	x_{1}^{2}+x_{2}^{2}+x_{3}^{2} & =b-\frac{a^{2}}{3}=:d,\label{eq:c1_2}\\
	x_{1}^{3}+x_{2}^{3}+x_{3}^{3} & =c-ab+\dfrac{9}{2}a^{3}=:e.\label{eq:c1_3}
	\end{align}
{We want to establish the conditions under which the
		system of equations has real roots}. After eliminating $x_{2}$ and
	$x_{3}$, we obtain 
	\begin{equation}
	6x_{1}^{3}-3dx_{1}-2e=0.\label{eq:cubic}
	\end{equation}
	This is a cubic equation for $x_{1}$ and always has a {real
		root} of the form $
	x_{1}^{*}=\sqrt[3]{\frac{e}{6}+\sqrt{\Delta_{3}}}+\sqrt[3]{\frac{e}{6}-\sqrt{\Delta_{3}}}$, 
	where $\Delta_{3}=(e/6)^{2}-(d/6)^{3}$ \citep{cardano1968}. Substitute
	$x_{1}^{*}$ in equation \eqref{eq:c1_1} and equation \eqref{eq:c1_2}.
	Then we get a quadratic equation for $x_{2}$, 
	\begin{equation}
	x_{2}^{2}+x_{1}^{*}x_{2}+x_{1}^{*2}-\frac{d}{2}=0.\label{eq:quadratic}
	\end{equation}
By the simple linear equation \eqref{eq:c1_1}, it
		is sufficient to analyse when the quadratic equation \eqref{eq:quadratic}
		has a real root. Using a standard result for quadratic equations,
		equation \eqref{eq:quadratic} has a real root if and only if the following
		discriminant is non-negative, 
	\[
	\Delta_{2}=x_{1}^{*2}-4\times(x_{1}^{*2}-\frac{1}{2}d)=2d-3x_{1}^{*2}
	\]
After some tedious algebra, the discriminant
	$\Delta_{2}$ can be simplified as follows: 
	\[
	\Delta_{2}=2d-3x_{1}^{*2}=-3\left(x_{1}^{*2}-\frac{2}{3}d\right)=-3\left(\sqrt[3]{\frac{e}{6}+\sqrt{\Delta_{3}}}-\sqrt[3]{\frac{e}{6}-\sqrt{\Delta_{3}}}\right)^{2}.
	\]
Now we can see that $\Delta_{2}\ge0$ if and only if $\Delta_{3}\le0$:
	\begin{enumerate}
		\item[(i)] When $\Delta_{3}=0$, $\Delta_{2}=-3\left(\sqrt[3]{\frac{e}{6}}-\sqrt[3]{\frac{e}{6}}\right)^{2}=0;$
		\item[(ii)] When $\Delta_{3}>0$, $\Delta_{2}=-3\left(\sqrt[3]{\frac{e}{6}+\sqrt{\Delta_{3}}}-\sqrt[3]{\frac{e}{6}-\sqrt{\Delta_{3}}}\right)^{2}<0;$
		\item[(iii)] When $\Delta_{3}<0$, there are three distinct real roots for \eqref{eq:cubic},
		which may be expressed in terms of polar coordinates $\cos\alpha\pm i\sin\alpha$,
		where $\alpha=\arccos\frac{e}{6R}$ and $R^{2}=\left(\frac{e}{6}\right)^{2}-\Delta_{3}=\left(\frac{d}{6}\right)^{3}$.
		So, after some algebra, $\Delta_{2}$ may be written: 
		\[
		\Delta_{2}
		=-3\left(\sqrt[3]{\frac{e}{6}+i\sqrt{-\Delta_{3}}}-\sqrt[3]{\frac{e}{6}-i\sqrt{-\Delta_{3}}}\right)^{2}=2d\sin^{2}\frac{a}{3}>0.
		\]
	\end{enumerate}
Recall that $\Delta_{3}=(e/6)^{2}-(d/6)^{3}$, hence,
	the necessary and sufficient condition for equation \eqref{eq:quadratic} to have a real root is $\Delta_{3}\le0$, viz. $d^{3}\ge6e^{2}.$ The proof is complete.
\end{proof}
\begin{proof}[Proof of Theorem~\ref{theorem HPSW} (ii) and (iii)]
	\label{app:proof first column} Here we establish
	the necessary and sufficient conditions for the system \eqref{eq:kollo:p}-\eqref{eq:kollo:sum}
	to be solvable when $k=1,\dots,n$. The system can be expressed in
	matrix form, as in equations \eqref{eq:linear_equations} and \eqref{eq:quadratic_equation},
	i.e. $\mathbf{U}\mathbf{y}=\mathbf{v}$ and $\mathbf{y}^{'}\mathbf{y}=m-\sum_{i=k+3}^{m}w_{ik}^{2}.$
	First, it is easy to see that condition (ii) is a necessary and sufficient
	condition for the linear equation {\eqref{eq:linear_equations}}
	to be solvable, i.e. $\mbox{Rank}\left(\mathbf{U}\right)=\mbox{Rank}\left(\left[\mathbf{U},\mathbf{v}\right]\right)$
	where $\left[\mathbf{U},\mathbf{v}\right]$ is the augmented matrix.
	Next we need to focus on the quadratic equation \eqref{eq:quadratic_equation}.
The idea is to use the linear equation \eqref{eq:linear_equations}
		to identify a set of linearly independent columns $\mathbf{U}_{1}$
		of $\mathbf{U}$ with full column rank, which enables us to express
		the rest columns $\mathbf{U}_{2}$ as a linear combination of $\mathbf{U}_{1}$,
		then the quadratic equation \eqref{eq:quadratic_equation} becomes
		an equation on $\mathbf{U}_{1}$. More specifically, re-write equation \eqref{eq:linear_equations}
	using $\mathbf{U}_{1}$ and $\mathbf{U}_{2}$ as
	$\mathbf{U}_{1}\mathbf{y}_{1}+\mathbf{U}_{2}\mathbf{y}_{2}=\mathbf{v}$,	where the vectors $\mathbf{y}_{1}$ and $\mathbf{y}_{2}$ contain
	the elements in $\mathbf{y}$ that correspond to $\mathbf{U}_{1}$
	and $\mathbf{U}_{2}$ respectively. Therefore, we could eliminate
	$\mathbf{y}_{1}$ in equation \eqref{eq:quadratic_equation} by $\mathbf{y}_{1}=\mathbf{U}_{1}^{+}\left(\mathbf{v}-\mathbf{U}_{2}\mathbf{y}_{2}\right)$,
	where $\mathbf{U}_{1}^{+}=\left(\mathbf{U}_{1}^{'}\mathbf{U}_{1}\right)^{-1}\mathbf{U}_{1}^{'}$
		is a Moore-Penrose inverse of $\mathbf{U}_{1}$ \citep{golub2012}.
		The Moore-Penrose inverse is in fact unique in this case, as $\mathbf{U}_{1}$
		has the full column rank $\mathbf{U}_{1}$ of $\mathbf{U}$. 
		Then
	the quadratic equation \eqref{eq:quadratic_equation} becomes 
	\begin{equation}
	\mathbf{y}_{2}^{'}\mathbf{G}\mathbf{y}_{2}-2\mathbf{y}_{2}^{'}\mathbf{g}+\mathbf{v}^{'}\left(\mathbf{U}_{1}\mathbf{U}_{1}^{'}\right)^{+}\mathbf{v}+\sum_{j=i+3}^{m}w_{ji}^{2}-m=0.\label{eq:quadratic_equation3}
	\end{equation}
	where $\mathbf{G}=\mathbf{I}+\mathbf{U}_{2}^{'}\left(\mathbf{U}_{1}\mathbf{U}_{1}^{'}\right)^{+}\mathbf{U}_{2}$
	and $\mathbf{g}=\mathbf{U}_{2}^{'}\left(\mathbf{U}_{1}\mathbf{U}_{1}^{'}\right)^{+}\mathbf{v}$.
	After completing the square, equation \eqref{eq:quadratic_equation3}
	becomes: 
	\begin{equation*}
	\mathbf{z}^{'}\mathbf{G}\mathbf{z}=\mathbf{g}^{'}\mathbf{G}^{-1}\mathbf{g}-\mathbf{v}^{'}\left(\mathbf{U}_{1}\mathbf{U}_{1}^{'}\right)^{+}\mathbf{v}-\sum_{j=i+3}^{m}w_{ji}^{2}+m,\label{eq:quadratic_equation4}
	\end{equation*}
	where $\mathbf{z}=\mathbf{y}_{2}-\mathbf{G}^{-1}\mathbf{g}$. Notice
		that $\mathbf{G}$ is a positive-definite matrix, hence to ensure
	real roots of the above equation the right side of the equation must
	be non-negative. Therefore we obtain the condition (iii), viz. 
	\[
	\mathbf{g}^{'}\mathbf{G}^{-1}\mathbf{g}-\mathbf{v}^{'}\left(\mathbf{U}_{1}\mathbf{U}_{1}^{'}\right)^{+}\mathbf{v}\,\ge\sum_{j=i+3}^{m}w_{ji}^{2}-m.
	\]
\end{proof}
\begin{proof}[Proof of Corollary~\ref{lem:zero arbitrary values}]
	\label{app:proof zero arbitrary values} We show that the zero arbitrary values are admissible under very
		mild conditions. In fact, they are not admissible only
		in the highly unusual cases where the number of simulations is very
	small and there exists one element of the Kollo
	skewness vector with large magnitude. To
		proceed, we check whether the system \eqref{eq:kollo:p}--\eqref{eq:kollo:sum}
		satisfy the conditions in use Theorem 1, column by column, when all
		arbitrary values are zero. When $k=1$ we just need to check whether
	zero arbitrary values satisfy condition Theorem 1 (i). In this case,
	we have 
	\[
	a=-\sum_{i=4}^{m}w_{i1}=0,\,b=m-\sum_{i=4}^{m}w_{i1}^{2}=m,\,\mbox{and}\,c=mp_{1}-\sum_{i=4}^{m}w_{i1}^{3}=mp_{1}.
	\]
	Then the equation \eqref{eq:first_realroots} becomes: 
	\begin{equation}
	\,(b-\frac{a^{2}}{3})^{3}\ge6(c-ab+\frac{9}{2}a^{3})^{2}\Leftrightarrow\ b^{3}\ge6c^{2}\Leftrightarrow m^{3}\ge6m^{2}p_{1}^{2}\Leftrightarrow p_{1}^{2}\le\frac{m}{6}.\label{eq:app:zero_1}
	\end{equation}
	In practical applications, $\frac{m}{6}$ is much larger than $p_{1}^{2}$,
	so equation \eqref{eq:app:zero_1} holds. 
	For $k=2,\dots,n$ we need to check whether conditions Theorem 1
	(ii)--(iii) are satisfied. Recall that in this case $\mathbf{U}$
	and $\mathbf{v}$ become 
	\[
	\mathbf{U}=\begin{bmatrix}s_{11}^{2} & s_{21}^{2} & s_{31}^{2} & 0 & 0 & \dots & 0 & 0\\
	s_{11} & s_{21} & s_{31} & 0 & 0 & \dots & 0 & 0\\
	s_{12} & s_{22} & s_{32} & s_{42} & 0 & \dots & 0 & 0\\
	\vdots & \vdots & \vdots & \vdots & \vdots & \vdots & \vdots & \vdots\\
	s_{1,k-1} & s_{2,k-1} & s_{3,k-1} & s_{4,k-1} & s_{5,k-1} & \dots & s_{k+1,,k-1} & 0\\
	1 & 1 & 1 & 1 & 1 & \dots & 1 & 1
	\end{bmatrix},\quad\mathbf{v}=\begin{bmatrix}mp_{k}\\
	0\\
	0\\
	\dots\\
	0\\
	0
	\end{bmatrix}.
	\]
	\\
Let $\mathbf{U}_{1}$ be the the matrix consisting
	of the first $k+1$ columns of $\mathbf{U}$, then we can compute
	the determinant of $\mathbf{U}_{1}\mathbf{U}_{1}^{'}$ as: 
	\[
	\mbox{det}(\mathbf{U}_{1}\mathbf{U}_{1}^{'})=t(k+1)m^{k-1},
	\]
	where $t=s_{11}^{4}+s_{21}^{4}+s_{31}^{4}-m\left(\frac{m}{k+1}+p_{1}^{2}+\dots+p_{k-1}^{2}\right)$.
	As $t>0$ by assumption, we get that $\mathbf{U}_{1}$ is invertible,
	hence is of full row rank, which implies that $\mathbf{v}$ must be
	in the span of $\mathbf{U}_{1}$. Hence Theorem 1 (ii) holds. 

If the system has real solutions, the condition (iii) requires $\mathbf{g}^{'}\mathbf{G}^{-1}\mathbf{g}-\mathbf{v}^{'}\left(\mathbf{U}_{1}\mathbf{U}_{1}^{'}\right)^{+}\mathbf{v}\ge-m$. After some algebraic computation, we get $$\mathbf{v}^{'}\left(\mathbf{U}_{1}\mathbf{U}_{1}^{'}\right)^{+}\mathbf{v} =\frac{m^{2}p_{k}^{2}}{t}, \,\,
	\mathbf{G}
	 =1+\frac{1}{k+1}+\frac{1}{t}\left(\frac{m}{k+1}\right)^{2}, \,\, \mathbf{g}
	 =-\frac{m^{2}p_{k}}{t(k+1)},$$ and condition (iii) becomes
	$\frac{t}{m}+\frac{m}{(k+1)(k+1)}\ge p_{k}^{2}.\label{eq:app:zero_k}$
	In practical applications, $m$ is much larger than $n$ and $k\le n$,
	so this inequality always holds. Thus the
		proof is complete. To summarise, setting all arbitrary
		values to zero is admissible under very mild conditions. 
\end{proof}

\subsection{Further Results on Normal Arbitrary Values}
\def\theequation{B.\arabic{equation}}
\setcounter{equation}{0}

\label{app:normal} Here we discuss the theoretical failure rate for
the first column of $\mathbf{S}_{mn}$, that is $\mathbf{s}_{1}=[s_{11},\dots,s_{m1}]^{'}$, 
when using normal arbitrary values $s_{i1}=w_{i1}$ for $i=4,\dots,m$
which are drawn from $\mathcal{N}(0,\sigma^{2})$ with the mean and
variance adjustment given by equation \eqref{eq:adjust}. This way,
the arbitrary variables should follow i.i.d. $\mathcal{N}(0,\sigma^{2})$,
and we can denote these $m-3$ random variables  $\left[W_{41},W_{51},...,W_{m1}\right]^{'}$.
We can calculate auxiliary   variables corresponding to their first three moments:
\begin{align*}
	M_{1} & =\frac{1}{m-3}\sum_{i=4}^{n}\left(W_{i1}\right)=\frac{1}{m-3}\sum_{i=4}^{n}\left(\frac{\sigma}{s_{z}}\left(Z_{j}-\bar{Z}\right)\right)=0\\
	M_{2} & =\frac{1}{m-3}\sum_{i=4}^{n}\left(W_{i1}-M_{1}\right)^{2}=\frac{1}{m-3}\sum_{i=4}^{n}\left(\frac{\sigma}{s_{z}}\left(Z_{j}-\bar{Z}\right)\right)^{2}=\sigma^{2}\\
	M_{3} & =\frac{1}{m-3}\sum_{i=4}^{n}\left(W_{i1}-M_{1}\right)^{3}=\frac{1}{m-3}\sum_{i=4}^{n}\left(\frac{\sigma}{s_{z}}\left(Z_{j}-\bar{Z}\right)\right)^{3}=\frac{\sum_{i=4}^{n}\left(Z_{j}-\bar{Z}\right)^{3}}{(m-3)s_{z}^{3}}\sigma^{3}
\end{align*}
Note that $M_{1}$ and $M_{2}$ are both constants, taking the values 0 and
$\sigma^{2}$ respectively. 

By the central
limit theorem \citep{joanes1998}  when the simulation sample size $m$ is large   $M_{3}$ has an approximate normal distribution with mean 0 and variance
$\frac{6(m-5)}{m(m-2)}\sigma^{6}$. Therefore,
for any arbitrary values selected this way, we have: 
$$\sum_{i=4}^{m}w_{i1}=0, \,\sum_{i=4}^{m}w_{i1}^{2}=(m-3)\sigma^{2}, \,\sum_{i=4}^{m}w_{i1}^{3}=(m-3)M_{3}.$$
Next we check Lemma \ref{lemma 1} for  
$\left[W_{41},W_{51},...,W_{m1}\right]^{'}$ to see whether
they are admissible. Since $a=0$, $b=m-(m-3)\sigma^{2}$ and
$c=mp_{1}-(m-3)M_{3}$ in equation \eqref{eq:first_realroots} we
can rewrite the condition as:
\begin{equation}
	p_{1}q-\sqrt{\frac{m-3}{6}(q-\sigma^{2})^{3}}\le M_{3}\le p_{1}q+\sqrt{\frac{m-3}{6}(q-\sigma^{2})^{3}}\label{eq:app:interval of M3}
\end{equation}
where $q=\frac{m}{m-3}$. So the probability of $\left[W_{41},W_{51},...,W_{m1}\right]^{'}$
being admissible is the probability of $M_3$ lying within the range given by \eqref{eq:app:interval of M3}.

Let $h(\sigma)$ denote the failure rate for the arbitrary values,
i.e. the probability that the arbitrary values do not satisfy the
condition \eqref{eq:app:interval of M3}. When $\sigma=0$, $\left[W_{41},W_{51},...,W_{m1}\right]^{'}$
are essentially  zero. By Corollary  \ref{lem:zero arbitrary values} zero arbitrary values are
always admissible when $p_{1}^{2}\le\frac{m}{6}$, hence
we have $h(0)=0$. When $\sigma\neq0$, because $M_{3}$ is approximately
normal distributed with mean 0 and variance $\frac{6(m-5)}{m(m-2)}\sigma^{6}$,
let $\Phi(\cdot)$ and $\varphi(\cdot)$ denote the cumulative distribution
function and probability density function for a standard normal distribution,
then we have 
\begin{equation}
	h(\sigma)=1-\left(\Phi(I_{U})-\Phi(I_{L})\right)\qquad\quad\text{for }0<\sigma\le\sqrt{q}\label{eq:app:P}
\end{equation}
where $I_{U}=\frac{1}{\sigma^{3}}\sqrt{\frac{m(m-2)}{6(m-5)}}\left(p_{1}q+\sqrt{\frac{m-3}{6}\left(q-\sigma^{2}\right)^{3}}\right)$ 
and $I_{L}=\frac{1}{\sigma^{3}}\sqrt{\frac{m(m-2)}{6(m-5)}}\left(p_{1}q-\sqrt{\frac{m-3}{6}\left(q-\sigma^{2}\right)^{3}}\right)$.

To illustrate this result, Figure \ref{fig: success first} depicts theoretical failure rates when $\sigma^2 \in \left[0, \frac{m}{m-3}\right]$. for $m=50$ (above) and $m=1000$ (below) and for various value of rotated Kollo skewness between 0 and 11.2, as shown in the legend.  All curves are monotonic increasing from 0 to 1 on the interval $\sigma^2 \in\left[ 0,\frac{m}{m-3}\right]$ and there is an obvious jump in each curve when either $m$ or $|p_1|$ are large. Comparing the cases $m=50$ and  $m=1000$, for any given given failure rate $\alpha$, a larger sample size can allow for larger $|p_1|$ with the same $\sigma^2$, or alternatively a greater $\sigma^2$ for the same $|p_1|$. Therefore, when the target Kollo skewness is far away from $\mathbf{0}_n$ we need to consider larger sizes of simulations and smaller $\sigma^2$ to ensure a reasonably low failure rate.

\begin{figure}[H]
		\caption{\footnotesize{\textbf{The Theoretical Failure Rate of the First Column with Different $p_{1}$.} \\The upper plots (a) are for $m=50$ and those below (b) are for $m=1000$. The arbitrary values in the first column are chosen from $\mathcal{N}(0,\sigma^{2})$. The red and blue horizontal dashed line depict the failure rate at 5\% and 10\%.}}
	\label{fig: success first}
	\centering
	\subfigure[$m=50$]{
		\includegraphics[width=13cm,height=4cm,trim={2cm 0.5cm 2.5cm 1cm clip}]{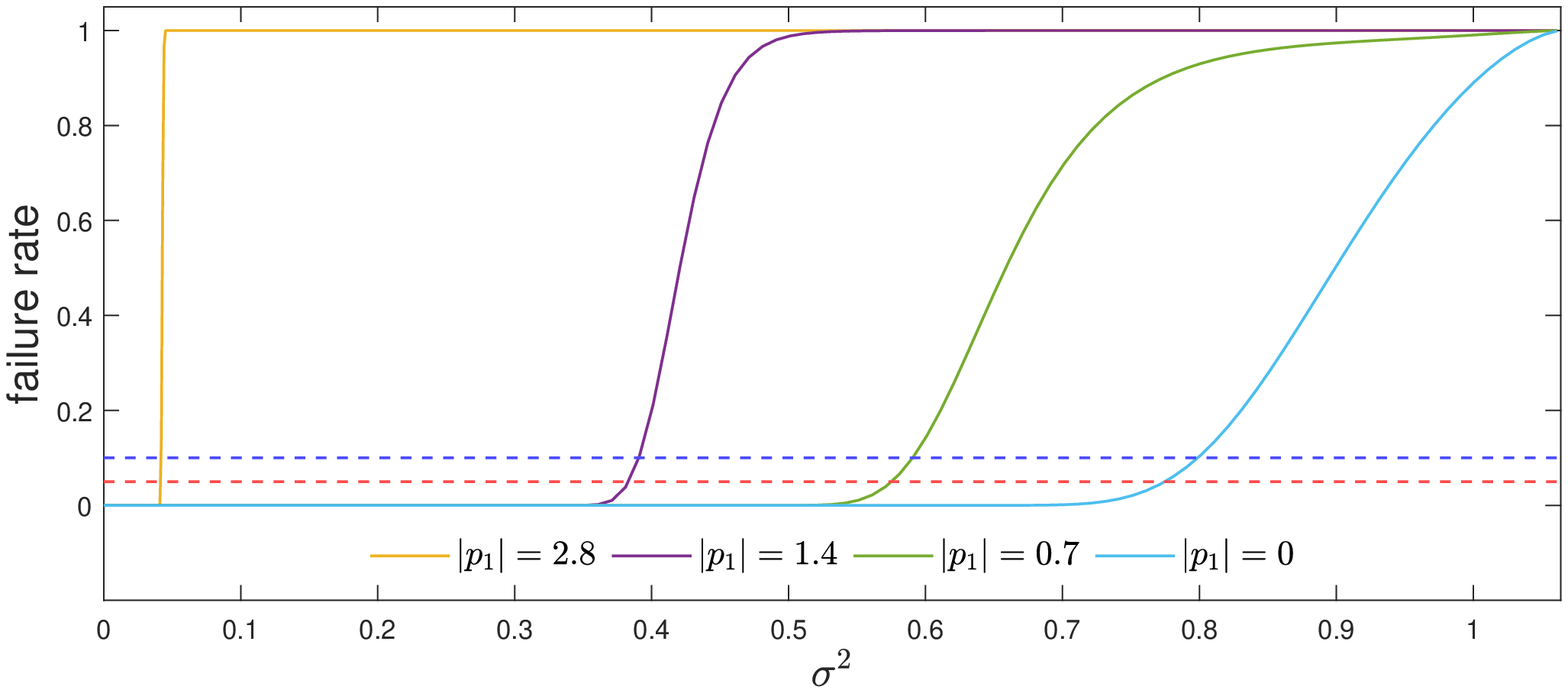}
	}
	\\
	\subfigure[$m=1000$]{
		\includegraphics[width=13cm,height=4cm,trim={2cm 0.5cm 2.5cm 0.5cm clip}]{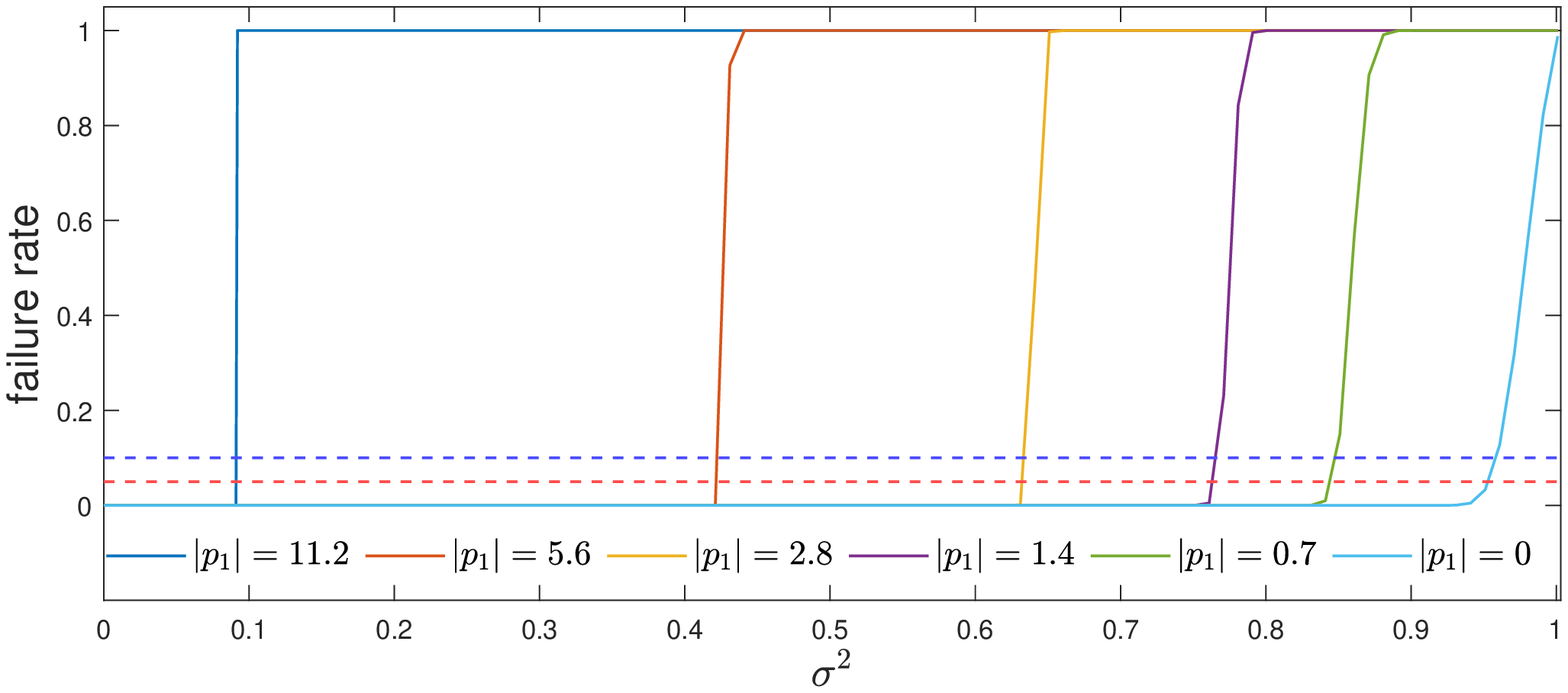}
	}
\end{figure}

\subsection{Attainable Skewness}\label{app: attainable p1}
This section derives new formulae for the range of skewness that can be attained using some well-known univariate distributions, and then computes this range when the distribution is standardized to have zero mean and unit variance. We provide three examples:

\begin{example}[SN distribution]
	The SN distribution is three-parameter distribution with PDF 
	\[
	f(t)=\frac{2}{\Omega}_{n}\phi\left(\frac{t-\xi}{\Omega}_{n}\right)\Phi\left(\alpha\left(\frac{t-\xi}{\Omega}_{n}\right)\right),
	\]
	where where $\xi$, $\omega>0$, $\alpha$ are parameters; $\phi$ and $\Phi$ denote the PDF and CDF for the standard Gaussian distribution with zero mean and unit variance. We can obtain the constraints OF zero mean, unit variance and $p_1$ skewness,
	\begin{equation}
	\begin{cases}
	\begin{aligned}\xi+\omega\delta\sqrt{\frac{2}{\pi}} & =0 & \qquad\qquad\text{(mean)}\\
	\omega^{2}\left(1-\frac{2\delta^{2}}{\pi}\right) & =1 & \qquad\qquad\text{(variance)}\\
	\dfrac{\sqrt{2}(4-\pi)\delta^{3}}{(\pi-2\delta^{2})^{3/2}} & =p_{1} & \qquad\qquad\text{(skewness)}
	\end{aligned}
	\end{cases}\label{eq: SN}
	\end{equation}
	where $\delta=\dfrac{\alpha}{\sqrt{\left(1+\alpha^{2}\right)}}$.  The skewness $\dfrac{\sqrt{2}(4-\pi)\delta^{3}}{(\pi-2\delta^{2})^{3/2}}$
	of the SN distribution is a function of $\alpha$, by taking its derivative
	with respect to $\alpha$, it is not difficult to show that it is
	monotone increasing with respect to $\alpha$. Hence by passing $\alpha$
	to $\pm\infty$, we can show that the range of the skewness is $(-0.995,0.995)$.
	This implies that $\forall p_{1}\in(-0.995,0.995)$, there always
	exists $\alpha$ such that the skewness of $Z_{1}$ matches $p_{1}$.
	Once $\alpha$ is obtained according to skewness $p_{1}$, we can
	easily get $\omega=\frac{1}{1-\frac{2\delta^{2}}{\pi}}$ and $\xi=-\omega\delta\sqrt{\frac{2}{\pi}}$
	by considering the constraints for mean and variance in equation~\eqref{eq: SN}.
	Therefore, the attainable values of $p_{1}$ are indeed $(-0.995,0.995)$. 
\end{example}
\begin{example}[NIG distribution]
	The NIG distribution is a four-parameter distribution, with PDF 
	\[
	f(t)=\frac{\alpha\delta K_{1}\left(\alpha\sqrt{\delta^{2}+(t-\mu)^{2}}\right)}{\pi\sqrt{\delta^{2}+(t-\mu)^{2}}}\exp(\delta\gamma+\beta(t-\mu)),
	\]
	where $K_{1}$ denotes a modified Bessel function of the third kind;	$\mu,$$\alpha$, $\beta(\beta<\alpha)$ and $\delta$ are parameters;	and $\gamma=\sqrt{\alpha^{2}-\beta^{2}}$. We can obtain the following moment constraints:
	\[
	\begin{cases}
	\begin{aligned}\mu+\frac{\delta\beta}{\gamma} & =0 & \quad\qquad\qquad\qquad\text{(mean)}\\
	\frac{\delta\alpha^{2}}{\gamma^{3}} & =1 & \quad\qquad\qquad\qquad\text{(variance)}\\
	\frac{3\beta}{\alpha\sqrt{\delta}\gamma} & =p_{1} & \quad\qquad\qquad\qquad\text{(skewness)}
	\end{aligned}
	\end{cases}
	\]
	Solving the constraints for mean and variance in the above equation,
	we can obtain $\delta=\frac{\gamma^{3}}{\alpha^{2}}$ and $\mu=-\frac{\delta\beta}{\gamma}$.
	Hence the constraint for skewness becomes $\frac{3\beta}{\alpha^{2}-\beta^{2}}=p_{1}$,
	whose solution always exists, because $\forall p_{1}\in\mathbb{R}$,
	we can choose any $\beta=\frac{p_{1}}{3}$ and $\alpha=\sqrt{3\beta+p_{1}\beta^{2}}.$
	Hence the attainable values of $p_{1}$ are $\mathbb{R}$.
\end{example}
\begin{example}
	The four-parameter Beta distribution has PDF 
	\[
	f(t;\alpha,\beta,b,c)=\frac{(t-b)^{\alpha-1}(c-t)^{\beta-1}}{(c-b)^{\alpha+\beta-1}B(\alpha,\beta)},
	\]
	where $\alpha>0,\beta>0,b,c(c>b)$ are parameters; ${\displaystyle \mathrm{B}(\alpha,\beta)=\frac{\Gamma(\alpha)\Gamma(\beta)}{\Gamma(\alpha+\beta)}};$
	and $\Gamma$ is the Gamma function. We can obtain the following moment constraints,
	\begin{equation}
	\begin{cases}
	\begin{aligned}\frac{\alpha c+\beta b}{\alpha+\beta} & =0 & \quad\quad\text{(mean)}\\
	\frac{\alpha\beta(c-b)^{2}}{(\alpha+\beta)^{2}(\alpha+\beta+1)} & =1 & \quad\quad\text{(variance)}\\
	\frac{2(\beta-\alpha)\sqrt{\alpha+\beta+1}}{(\alpha+\beta+2)\sqrt{\alpha\beta}} & =p_{1} & \quad\quad\text{(skewness)}
	\end{aligned}
	\end{cases}\label{eq: beta distribution}
	\end{equation}
	Notice that if we fix any $\beta>0$, $\frac{2(\beta-\alpha)\sqrt{\alpha+\beta+1}}{(\alpha+\beta+2)\sqrt{\alpha\beta}}\rightarrow\infty$
	as $\alpha\rightarrow0$; if we fix any $\alpha>0$, $\frac{2(\beta-\alpha)\sqrt{\alpha+\beta+1}}{(\alpha+\beta+2)\sqrt{\alpha\beta}}\rightarrow-\infty$
	as $\beta\rightarrow0$; and $\frac{2(\beta-\alpha)\sqrt{\alpha+\beta+1}}{(\alpha+\beta+2)\sqrt{\alpha\beta}}=0$
	if $\alpha=\beta$. Hence $\forall p_{1}\in\mathbb{R}$, there exists
	$\alpha$ and $\beta$ such that skewness $p_{1}$. Next, we obtain
	the equations for $b$ and $c$ by solving the constraints for mean
	and variance in equation~\eqref{eq: beta distribution} as follows,
	\begin{align*}
	\begin{cases}
	\alpha c+\beta b & =0\\
	\ \ c-\ b & =\frac{(\alpha+\beta)(\sqrt{\alpha+\beta+1})}{\sqrt{\alpha\beta}},
	\end{cases}
	\end{align*}
	which always has a solution, because the matrix $\begin{bmatrix}\alpha & \beta\\
	1 & -1
	\end{bmatrix}$ is invertible for $\alpha,\beta>0$. Hence the attainable values
	of $p_{1}$ are $\mathbb{R}$.
\end{example}
Table \ref{tab: skewness}  summarizes the results and computes attainable ranges for distributions with zero mean and unit variance:
\begin{table}[H]
	\centering
	\footnotesize
	\caption{\footnotesize{\textbf{Attainable Skewness for Well-Known Distributions.}}}
	\begin{tabular}{llcccc}
		\hline \\[-8pt]
		\multirow{2}[0]{*}{Distribution} & \multirow{2}[0]{*}{PDF}   & \multirow{2}[0]{*}{Mean}  & \multirow{2}[0]{*}{Variance} & \multirow{2}[0]{*}{Skewness} & Attainable\\
		&&&&& skewness \\[1pt]
		\hline\\[-7pt]
		N& $\frac{1}{\sigma \sqrt{2\pi}}e^{-\frac{1}{2}\left(\frac{t-\mu}{\sigma}\right)^2 }$ & \multirow{2}[0]{*}{$\mu$} & \multirow{2}[0]{*}{$\sigma^2$} & \multirow{2}[0]{*}{0} & \multirow{2}[0]{*}{0} \\
		$f(t;\mu,\sigma)$ & {\scriptsize($\mu \in \mathbb{R}$, $ \sigma>0$)}  &   &    &       &       
		\\[2pt]
		\hline
		\\[-7pt]
		T& $\frac{\Gamma\left(\frac{\nu+1}{2}\right)}{\Gamma\left(\frac{\nu}{2}\right)\sqrt{\pi\nu}\sigma}\left(1+ \frac{1}{\nu}\left(\frac{t-\mu}{\sigma}\right)^{2}\right)^{-\frac{\nu+1}{2}}$     & \multirow{2}[0]{*}{$\mu$} & \multirow{2}[0]{*}{$\frac{\nu}{\nu-2}\sigma^2$} & \multirow{2}[0]{*}{0}& \multirow{2}[0]{*}{0} \\
		$f(t;\mu,\sigma,\nu)$& {\scriptsize($\mu \in \mathbb{R}$, $\nu, \sigma>0$ )} & &  &       &       
		\\[2pt]
		\hline
		\\[-7pt]
		SN& $\frac{2}{\Omega}_{n}\phi\left(\frac{t-\xi}{\Omega}_{n}\right)\Phi\left(\alpha\left(\frac{t-\xi}{\Omega}_{n}\right)\right)$& 
		\multirow{2}[0]{*}{$\xi+\omega\delta\sqrt{\frac{2}{\pi}}$} & \multirow{2}[0]{*}{$\omega^{2}\left(1-\frac{2\delta^{2}}{\pi}\right)$} &\multirow{2}[0]{*}{ $\dfrac{\sqrt{2}(4-\pi)\delta^{3}}{(\pi-2\delta^{2})^{3/2}} $} &\multirow{2}[0]{*}{$(-0.995,0.995)$}\\
		$f(t;\xi,\omega,\alpha)$& {\scriptsize($\xi,\alpha \in \mathbb{R}$, $\omega>0$, $\delta=\alpha/\sqrt{1+\alpha^{2}}$)} &       &       & &    
		\\[2pt]
		\hline
		\\[-7pt]
		NIG & $\frac{\alpha\delta K_{1}\left(\alpha\sqrt{\delta^{2}+(t-\mu)^{2}}\right)}{\pi\sqrt{\delta^{2}+(t-\mu)^{2}}}e^{\delta\gamma+\beta(t-\mu)}$ & \multirow{2}[0]{*}{$\mu+\frac{\delta\beta}{\gamma^3} $} & \multirow{2}[0]{*}{$\frac{\delta\alpha^{2}}{\gamma^{3}}$} & \multirow{2}[0]{*}{$\frac{3\beta}{\alpha\sqrt{\delta}\gamma}$}& \multirow{2}[0]{*}{$\mathbb{R} $}\\
		$f(t;\alpha,\beta,\delta,\mu)$& {\scriptsize($\alpha,\beta,\delta,\mu \in \mathbb{R}$, $|\alpha|\ge|\beta|$)}    & & &       &   
		\\[2pt]
		\hline
		\\[-7pt]
		Beta  & $\frac{(t-b)^{\alpha-1}(c-t)^{\beta-1}}{(c-b)^{\alpha+\beta-1}B(\alpha,\beta)} $ & \multirow{2}[0]{*}{$\frac{\alpha c+\beta b}{\alpha+\beta}$} & \multirow{2}[0]{*}{$\frac{\alpha\beta(c-b)^{2}}{(\alpha+\beta)^{2}(\alpha+\beta+1)} $} & \multirow{2}[0]{*}{$\frac{2(\beta-\alpha)\sqrt{\alpha+\beta+1}}{(\alpha+\beta+2)\sqrt{\alpha\beta}}\in \mathbb{R} $}& \multirow{2}[0]{*}{$\mathbb{R} $} \\
		$f(t;\alpha,\beta,b,c)$& {\scriptsize($\alpha,\beta,b,c \in \mathbb{R}$, $c>b$) }&       &       &      &
		\\[2pt]
		\hline
	\end{tabular}%
	\label{tab: skewness}%
\end{table}%

\end{document}